\documentclass[journal]{IEEEtran}
\usepackage{cite}  
\usepackage{graphicx}
\usepackage{algorithm}  
\usepackage{algpseudocode}  
\usepackage{latexsym}
\usepackage{amsmath,amssymb,amsfonts}  
\usepackage{amsthm}  
\usepackage{bm}  
\usepackage{xcolor}
\usepackage[acronym,shortcuts]{glossaries}
\usepackage[font=footnotesize]{subcaption}
\usepackage[font=footnotesize]{caption}
\usepackage{makecell}
\newacronym{MIMO}{MIMO}{multiple-input-multiple-output}
\newacronym{XL-MIMO}{XL-MIMO}{extremely large-scale MIMO}
\newacronym{5G}{5G}{fifth-generation}
\newacronym{mmWave}{mmWave}{millimeter-wave}
\newacronym{cmWave}{cmWave}{centimeter-wave}
\newacronym{BS}{BS}{base station}
\newacronym{UE}{UE}{user equipment}
\newacronym{MF}{MF}{matched-filter}
\newacronym{ZF}{ZF}{zero-forcing}
\newacronym{LMMSE}{LMMSE}{linear minimum mean square error}
\newacronym{MSE}{MSE}{mean square error}
\newacronym{MPA}{MPA}{message passing algorithm}
\newacronym{AMP}{AMP}{approximate message passing}
\newacronym{GAMP}{GAMP}{generalized AMP}
\newacronym{VAMP}{VAMP}{vector AMP}
\newacronym{OAMP}{OAMP}{orthogonal AMP}
\newacronym{GaBP}{GaBP}{Gaussian belief propagation}
\newacronym{EP}{EP}{expectation propagation}
\newacronym{MF-EP}{MF-EP}{matched-filter EP}
\newacronym{EC}{EC}{expectation consistent}
\newacronym{OvEP}{OvEP}{overlapping block partitioning EP}
\newacronym{NOvEP}{NOvEP}{non-overlapping block partitioning EP}
\newacronym{SE}{SE}{state evolution}
\newacronym{AWGN}{AWGN}{additive white Gaussian noise}
\newacronym{PDF}{PDF}{probability density function}
\newacronym{FLOPs}{FLOPs}{floating point operations}
\newacronym{KL}{KL}{Kullback-Leibler}
\newacronym{ML}{ML}{maximum likelihood}
\newacronym{i.i.d.}{i.i.d.}{independent and identically distributed}
\newacronym{BER}{BER}{bit error rate}
\newacronym{QAM}{QAM}{quadrature amplitude modulation}
\newacronym{QPSK}{QPSK}{quadrature phase-shift keying}

\newtheorem{lemma}{Lemma}
\newtheorem{proposition}{Proposition}

\newcommand{\figsize}{0.58}

\begin{document}


\title{Expectation Propagation-Based Signal Detection \\ for Highly Correlated MIMO Systems}

\author{
Kabuto Arai,~\IEEEmembership{Graduate Student Member, IEEE},  
Takumi Yoshida,~\IEEEmembership{Graduate Student Member, IEEE},  \\
Takumi Takahashi,~\IEEEmembership{Member, IEEE},  
Koji Ishibashi,~\IEEEmembership{Senior Member, IEEE},
%
\thanks{
Kabuto Arai, Takumi Yoshida, and Koji Ishibashi are with the Advanced Wireless and Communication Research Center (AWCC), The University of Electro-Communications, Tokyo 182-8285, Japan (e-mail: \{k.arai, t.yoshida\}@awcc.uec.ac.jp, koji@ieee.org)}
\thanks{
Takumi Takahashi is with the Graduate School of Engineering, The University of Osaka, Yamada-oka 2-1, Suita 565-0871, Japan, (email: takahashi@comm.eng.osaka-u.ac.jp).
}
}

%
\markboth{Journal of \LaTeX\ Class Files,~Vol.~14, No.~8, August~2021}%
{Shell \MakeLowercase{\textit{et al.}}: Sample article using IEEEtran.cls for IEEE Journals}


\maketitle

\begin{abstract}
    Large-scale \acf{MIMO} systems typically operate in dense array deployments with limited scattering environments, leading to highly correlated and ill-conditioned channel matrices that severely degrade the performance of message-passing-based detectors.
    To tackle this issue, this paper proposes an \ac{EP}-based detector, termed \acf{OvEP}.
    In OvEP, the large-scale measurement vector is partitioned into partially overlapping blocks.
    For each block and its overlapping part, a low-complexity \acf{LMMSE}-based filter is designed according to the partitioned structure.
    The resulting LMMSE outputs are then combined to generate the input to the denoiser.
    In this combining process, subtracting the overlapping-part outputs from the block outputs effectively mitigates the adverse effects of inter-block correlation induced by high spatial correlation. 
    The proposed algorithm is consistently derived within the EP framework, and its fixed point is theoretically proven to coincide with the stationary point of a relaxed \acf{KL} minimization problem.
    The mechanisms underlying the theoretically predicted performance improvement are further clarified through numerical simulations.
    The proposed algorithm achieves performance close to conventional LMMSE-EP with lower computational complexity.
\end{abstract}

\glsresetall

\begin{IEEEkeywords}
    Expectation propagation (EP), 
    highly correlated channels, 
    multiple-input-multiple-output (MIMO).
\end{IEEEkeywords}

\IEEEpeerreviewmaketitle
\glsresetall

\section{Introduction}

\Ac{MIMO} systems, particularly massive \ac{MIMO}~\cite{2014Larsson_mMIMO} and \ac{XL-MIMO}~\cite{2024Wang_XL_MIMO_tutorial}, have emerged as pivotal technologies to meet the demand for enhanced spectral efficiency in next-generation wireless networks.
Densely deploying array antennas at a \ac{BS} enables the exploitation of the spatial degrees of freedom for multiplexing and diversity, thereby improving spectral efficiency~\cite{2004Tse_DM_tradeoff}.
To fully harness these abundant spatial degrees of freedom, accurate and efficient receiver design for \ac{MIMO} signal detection is one of the essential tasks.

Although the \ac{ML} detector achieves optimal detection performance, its computational cost scales exponentially with the number of BS antennas, which makes it impractical in large-scale systems~\cite{2019Albreem_mMIMO_detection}.
In contrast, linear detectors such as \ac{MF}, \ac{ZF}, and \ac{LMMSE} estimators reduce the computational burden but are inherently limited by the constraints of linear estimation.
To enhance detection performance with low computational complexity in large-scale systems, various \acp{MPA} have been proposed~\cite{2009Donoho_AMP, 2011Rangan_GAMP, 2003Kabashima_GaBP, 2015Meng_MFEP, 2019Rangan_VAMP, 2017Ma_OAMP}.
In general, \ac{MPA}-based approaches perform iterative updates by alternating between linear estimation, typically using an \ac{MF} or \ac{LMMSE} filter, and nonlinear inference via a denoiser function based on the prior distribution of the discrete signal, thereby progressively improving estimation accuracy.

One of the representative \acp{MPA} is \ac{AMP}~\cite{2009Donoho_AMP}.
Several variants, including \ac{GAMP}~\cite{2011Rangan_GAMP}, \ac{GaBP}~\cite{2003Kabashima_GaBP,2019Takahashi_ASB}, and MF-EP~\cite{2015Meng_MFEP}, have been proposed.
Since these algorithms employ the \ac{MF} without requiring matrix inversion during the iterative process, their computational complexity remains relatively low.
Their detection performance can achieve the Bayes-optimal performance in the large-system limit when the channel matrix follows zero-mean \ac{i.i.d.} Gaussian distribution (i.e., Rayleigh fading channels).
Furthermore, the asymptotic dynamics of the \ac{MSE} across iterations can be rigorously characterized using \ac{SE}~\cite{2011Bayati_AMP}.
Building on these theoretical foundations, \ac{AMP}-based low-complexity algorithms have been extensively investigated for applications in wireless communication systems~\cite{2018Mo_CE_ADC, 2014Wu_AMP_MIMO, 2022Takahahi_layerd, 2025Ueda_GF_AMP}.

However, the assumption of \ac{i.i.d.} Rayleigh channels is only valid under rich scattering environments with large angular spreads, numerous propagation paths, and sufficient antenna spacing at the BS~\cite{2005David_Fundamental}.
In practice, large-scale MIMO is typically deployed above the mid-band spectrum, including the \ac{cmWave} and \ac{mmWave} bands, where the shorter wavelengths allow the integration of large antenna arrays within limited physical space.
In such environments, the channel usually consists of a few dominant propagation paths with narrow angular spread due to severe path loss and blockage effects~\cite{2024Liu_model_CmMmST,2015Rappaport_mmWave_model}.
Furthermore, to accommodate a large number of antennas within restricted space, antenna elements are often placed with sub-wavelength spacing.
Consequently, the \ac{i.i.d.} Rayleigh channel assumption no longer holds in practical scenarios, and the resulting channel matrices become highly spatially correlated and ill-conditioned~\cite{2022Du_MIMO_correlation, 2015Gao_MIMO_correlation, 2022Dong_XL_correlation}.
Under these conditions, the convergence performance of \ac{AMP}-based algorithms degrades owing to the violation of the underlying assumptions~\cite{2019Rangan_AMP_damping}.

To improve detection performance across a wide range of channel matrices, alternative \acp{MPA}, namely \ac{VAMP} \cite{2019Rangan_VAMP} and \ac{OAMP} \cite{2017Ma_OAMP}, have been proposed.
Although developed independently through different methodologies, they are essentially the same algorithm within the Bayesian regime, also referred to as LMMSE-EP.
These algorithms can be derived from the \ac{EP}~\cite{2001Minka_EP} and \ac{EC}~\cite{2005Opper_EC,2016Fletcher_GeneralizedEC} frameworks with diagonally-restricted forms.
They can achieve Bayes-optimal performance in the large-system limit when the channel matrix is right-unitarily invariant, which includes zero-mean \ac{i.i.d.} Gaussian channels as a special case~\cite{2020Takeuchi_EP}.
Leveraging the property, these algorithms have been extensively applied to MIMO signal detection even under spatially correlated channels~\cite{2014Cespedes_EP_MIMO, 2021Ito_Bilinear_correlated, 2018He_OAMPNet, 2020He_OAMPNet2}.
However, they require full-dimensional matrix operations including multiplications and inversions, causing the computational complexity to increase significantly as the system size scales up.

To balance computational complexity and detection performance, an \ac{EP}-based algorithm has been proposed in~\cite{2020Wang_EP_subarray}, hereafter referred to as \emph{\ac{NOvEP}} in this paper.
In this algorithm, the large-scale measurement vector is partitioned into small-size non-overlapping blocks.
For each block, an \ac{LMMSE}-based filter produces tentative estimates, which are then combined to generate the input to the denoiser.
Subsequently, nonlinear estimation through the denoiser function generates the approximate posterior mean.
This iterative process progressively enhances estimation accuracy.
Since the \ac{LMMSE} filter is designed for each small-size block, \ac{NOvEP} significantly reduces the computational cost of matrix inversion compared to conventional \ac{LMMSE}-\ac{EP}.
Moreover, owing to block-wise decentralized processing that reduces the amount of data for baseband signals, this algorithm has been applied to efficient signal detection in \ac{XL-MIMO} systems~\cite{2020Wang_EP_subarray, 2025Arai_XL}, and cell-free MIMO systems~\cite{2021He_SubEP_CF,2025He_SubEP_CF}.
However, under highly correlated channels, the \ac{LMMSE} outputs across blocks also become strongly correlated.
Consequently, combining these correlated outputs may produce outliers in the denoiser input.
%
In iterative \acp{MPA}, such outliers appearing in the early iterations lead to error propagation in subsequent iterations~\cite{2022Tamaki_GAMP_feedback,2024Furudoi_ADD,2019Takahashi_ASB}.
Therefore, designing computationally efficient and accurate signal detection algorithms for highly correlated channels still remains a significant challenge.

To address these issues, this paper proposes an \ac{EP}-based detection algorithm, termed \ac{OvEP}.
The main contributions are summarized as follows:
\begin{itemize}
    \item 
        The proposed method partitions the large-scale measurement vector into small-size overlapping blocks, unlike the conventional \textit{NOvEP}~\cite{2020Wang_EP_subarray}.
        For each block, an \ac{LMMSE}-based filter is designed, and the resulting outputs are combined to generate the denoiser input. 
        In highly correlated channels, the inter-block correlation between the block-wise \ac{LMMSE} outputs becomes significant.
        To address this, the proposed method designs \ac{LMMSE} filters not only for the block parts but also for their overlapping parts. 
        By subtracting the \ac{LMMSE} outputs of the overlapping parts from those of the corresponding blocks, the inter-block correlation can be substantially reduced.
        Consequently, the proposed method improves detection performance compared with the conventional \textit{NOvEP}, with only a marginal increase in complexity.
    \item 
        The proposed algorithm is consistently derived within the \ac{EP} framework by exploiting the partitioned measurements in both block and overlapping parts.
        In the \ac{EP} framework, the approximate posterior, represented as a product of multiple Gaussian factors, is obtained by minimizing the \ac{KL} divergence.
        Since direct optimization over all factors is intractable, the original \ac{KL} minimization problem is decomposed into multiple sub-problems of local divergence minimization, each associated with a specific factor~\cite{2005minka_divergence}.
        These factors are then sequentially updated through alternating minimization of the local divergence.
        This derivation highlights the theoretical foundation of the proposed algorithm, distinguishing it from heuristic approaches.
    \item 
        In general, the alternating minimization of the local divergence does not necessarily minimize the original \ac{KL} divergence~\cite{2005minka_divergence}.
        Nevertheless, we theoretically prove that the fixed point of the proposed algorithm, derived through alternating optimization, corresponds to a stationary point of a relaxed version of the original \ac{KL} minimization.
    \item 
        To investigate the fundamental causes behind the performance improvement of the proposed algorithm, we provide both theoretical and empirical analyses.
        We verify that the \ac{MSE} of the denoiser input is determined by the inter-block correlation among the block-wise \ac{LMMSE} outputs.
        By subtracting the \ac{LMMSE} outputs of the overlapping parts, the proposed algorithm can reduce this inter-block correlation, thereby lowering the \ac{MSE} of the denoiser input. 
        Furthermore, reducing the inter-block correlation helps prevent the occurrence of outliers and keeps the denoiser from operating in its inactive region~\cite{2024Furudoi_ADD}, thereby mitigating error propagation during the iterative process. 
        These underlying mechanisms of the proposed algorithm are further validated and illustrated through numerical simulations.
\end{itemize}

\textit{Notation}:
For a matrix $\mathbf{A}$, $[\mathbf{A}]_{i,j}$ and $\Re [\mathbf{A} ]$ denote the $(i,j)$ entry and real part of $\mathbf{A}$. 
$(\cdot)^*$, $(\cdot)^\mathrm{T}$, $(\cdot)^\mathrm{H}$ denote conjugate, transpose, and Hermitian transpose. 
$\mathbf{I}_N$, $\mathbf{0}_M$, and $\mathbf{1}_M$ are the $N\times N$ identity matrix, the $M\times1$ zero vector, and the $M\times1$ all-ones vector. 
$\odot$ and $\oslash$ denote Hadamard (element-wise) product and division.
For a vector $\mathbf{x}$ and a matrix $\mathbf{A}$, $\mathbf{D}(\mathbf{x})$ and $\mathbf{d}(\mathbf{A})$ are the diagonal matrix constructed from $\mathbf{x}$ and the vector constructed from the diagonal elements of $\mathbf{A}$.
$\|\mathbf{x}\|^2_{\bm{\gamma}} \triangleq \sum_m \gamma_m |x_m|^2$ denotes the weighted $\ell_2$ norm. 
For an index set $\mathcal{N}$, $\mathbf{x}_{\mathcal{N}}$ and $\mathbf{A}_{\mathcal{N}}$ extract the rows of $\mathbf{x}$ and $\mathbf{A}$ indexed by $\mathcal{N}$.
For a PDF $p(x)$, $\mathbb{E}_{p(x)}[x]$ and $\mathbb{V}_{p(x)}[x]$ denote expectation and variance. 
$\int_{\mathbf{x}\setminus x_m} f(\mathbf{x})$ integrates $f(\mathbf{x})$ over all variables except $x_m$. 
$\mathcal{CN}(\bm{\mu},\mathbf{\Sigma})$ denote the circularly symmetric complex Gaussian distribution with mean $\bm{\mu}$ and covariance $\mathbf{\Sigma}$. 
$\mathrm{KL}(p||q) \triangleq \int p(x)\ln \tfrac{p(x)}{q(x)}$ and $\mathrm{H}(p) \triangleq -\int p(x)\ln p(x)$ are the KL divergence and differential entropy. 
For a functional $\mathcal{L}(q)$, $\delta \mathcal{L}/\delta q(\mathbf{x})$ denotes the functional derivative with respect to a function $q(x)$~\cite{2006Bishop_PRML}.
$\delta(\cdot)$ denotes the Dirac delta function.

\section{System Model}
\label{sec:system}

We consider an uplink \ac{MIMO} system, consisting of a \ac{BS} equipped with $N$ antennas and $M$ \acp{UE}, each with a single antenna.
The $m$-th UE transmits a discrete symbol $x_m \in \mathcal{X}$, independently drawn from a $Q$-\ac{QAM} constellation $\mathcal{X} \triangleq \{ \mathrm{x}_1, \ldots, \mathrm{x}_Q \}$.
Defining the transmit signal vector $\mathbf{x} \triangleq [x_1, \ldots, x_M]^\mathrm{T} \in \mathcal{X}^{M \times 1}$, the measurement vector at the BS, $\mathbf{y} \triangleq [y_1, \ldots, y_N]^\mathrm{T} \in \mathbb{C}^{N \times 1}$, can be expressed as
\begin{align}
    \label{eq:y}
    \mathbf{y} = \mathbf{H} \mathbf{x} + \mathbf{z},
\end{align}
where 
$\mathbf{H} \in \mathbb{C}^{N \times M}$ denotes the channel matrix between the BS and $M$ UEs, and 
$\mathbf{z}$ denotes an \ac{AWGN} vector following $\mathcal{CN}(\mathbf{0}_N, \sigma_z \mathbf{I}_N)$.
The final goal is to estimate the transmit vector $\mathbf{x}$ using the measurement vector $\mathbf{y}$ and the channel matrix $\mathbf{H}$.

Based on the linear measurement model in \eqref{eq:y}, the conditional \ac{PDF} $p(\mathbf{y}|\mathbf{x})$ and the prior distribution $p(\mathbf{x})$ are given by 
\begin{align}
    \label{eq:p_y}
    p(\mathbf{y}|\mathbf{x}) &= \prod_{n=1}^N p(y_n|\mathbf{x}) = \mathcal{CN}(\mathbf{y}| \mathbf{Hx}, \ \sigma_z \mathbf{I}_N), \\
    p(\mathbf{x}) =& \prod_{m=1}^M p(x_m) = \prod_{m=1}^M
    \frac{1}{Q} \sum_{\mathrm{x} \in \mathcal{X}} \delta(x_m - \mathrm{x}),
\end{align}
where the mean and covariance of $\mathbf{x}$ are given by $\mathbb{E}_{p(\mathbf{x})} [\mathbf{x}] = \mathbf{0}_M$ and $\mathbb{E}_{p(\mathbf{x})} [\mathbf{x} \mathbf{x}^\mathrm{H}] = \sigma_x \mathbf{I}_M$, respectively.

Using these PDFs, the posterior distribution is given by 
\begin{align}
    \label{eq:post}
    p(\mathbf{x}|\mathbf{y}) = p(\mathbf{y}|\mathbf{x}) p(\mathbf{x})  / p(\mathbf{y}), 
\end{align}
where $p(\mathbf{y}) = \int_{\mathbf{x}} p(\mathbf{y}|\mathbf{x}) p(\mathbf{x})$.
Since the posterior involves a multidimensional integration over the discrete-valued vector $\mathbf{x}$, exact computation is infeasible with reasonable complexity.
Therefore, this paper aims to approximate the posterior $p(\mathbf{x}|\mathbf{y})$ within a tractable complexity using the \ac{EP} framework.

\section{Proposed Method}

\begin{figure}[t]
    \centering
    \includegraphics[scale=0.3]{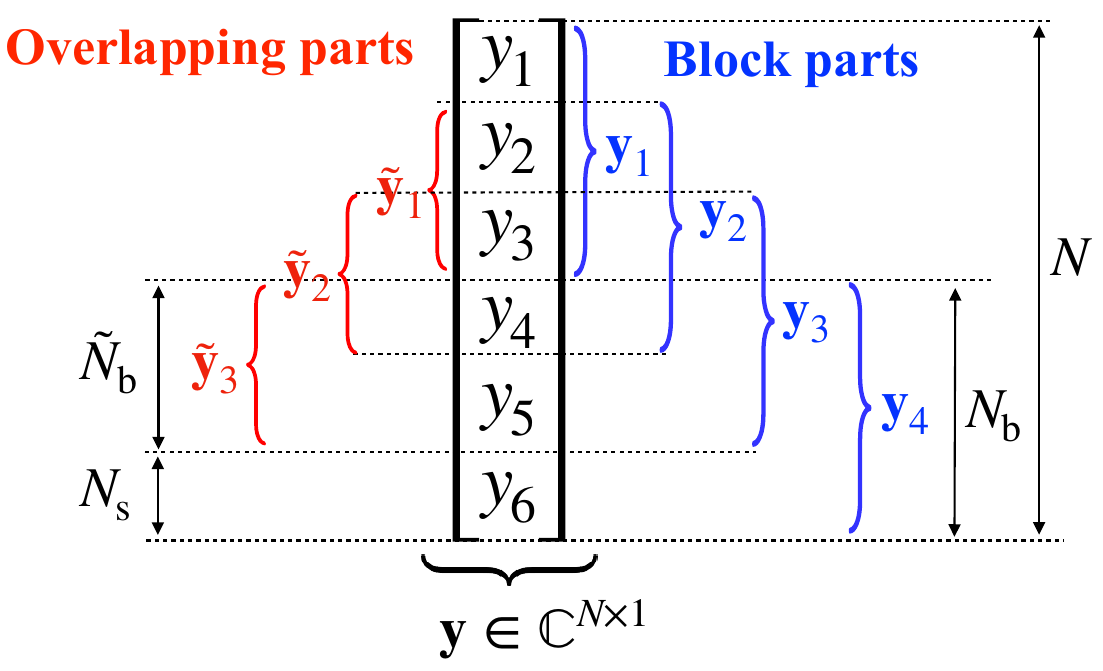}
    \caption{An example of block partitioning with $N=6,\ N_\mathrm{b}=3,\ N_\mathrm{s}=1,\ \tilde{N}_\mathrm{b}=2,$ and $L=4$.} 
    \label{fig:Overlaped_Block}
    \vspace{-3ex}
\end{figure}

\subsection{Measurement Partitioning}

As illustrated in Fig.~\ref{fig:Overlaped_Block}, the measurement vector $\mathbf{y}$ is partitioned into $L$ blocks, each obtained by sliding a window of size $N_\mathrm{b}$ with a shift of $N_\mathrm{s}$.
The number of blocks $L$ is given by $L=(N-N_\mathrm{b})/N_\mathrm{s} + 1$.
The block size $N_\mathrm{b}$ and the shift size $N_\mathrm{s}$ should be selected such that $1 \leq N_\mathrm{s} \leq N_\mathrm{b} \leq N$ and $L$ becomes an integer.
Defining the index set for the $l$-th block as 
\begin{align*}
    \mathcal{N}_l \triangleq \{(l-1) N_\mathrm{s} + 1, \ldots, (l-1) N_\mathrm{s} + N_\mathrm{b} \}, 
    \forall l \in \{1, \ldots, L \},
\end{align*}
the measurement vector for the $l$-th block can be expressed as $\mathbf{y}_l \triangleq \mathbf{y}_{\mathcal{N}_l} \in \mathbb{C}^{N_\mathrm{b} \times 1}$.
It should be noted that each block overlaps with its neighboring blocks by $\tilde{N}_\mathrm{b} \triangleq N_\mathrm{b} - N_\mathrm{s}$ elements.
Let $\tilde{\mathcal{N}}_l \triangleq \mathcal{N}_l \cap \mathcal{N}_{l+1}$ denote the index set for the overlapping elements between the $l$-th and $(l+1)$-th blocks, which is expressed as
\begin{align*}
    \tilde{\mathcal{N}}_l 
    = \{l N_\mathrm{s} + 1, \ldots, l N_\mathrm{s} + \tilde{N}_\mathrm{b} \},\ \forall l \in \{ 1, \ldots, L-1 \}.
\end{align*}
Then, the measurement vector for the $l$-th overlapping part is expressed as $\tilde{\mathbf{y}}_l \triangleq \mathbf{y}_{\tilde{\mathcal{N}}_l} \in \mathbb{C}^{\tilde{N}_\mathrm{b} \times 1}$.

Similarly, the channel matrix $\mathbf{H}$ is also partitioned into block parts $\{ \mathbf{H}_l \}_{l=1}^L$ and overlapping parts $\{ \tilde{\mathbf{H}}_l \}_{l=1}^{L-1}$, defined as  
$\mathbf{H}_l \triangleq \mathbf{H}_{\mathcal{N}_l} \in \mathbb{C}^{N_\mathrm{b} \times M}$ and 
$\tilde{\mathbf{H}}_l \triangleq \mathbf{H}_{\tilde{\mathcal{N}}_l} \in \mathbb{C}^{\tilde{N}_\mathrm{b} \times M}$.

An illustrative example of measurement partitioning is depicted in Fig.~\ref{fig:Overlaped_Block}, where $N=6,\ N_\mathrm{b}=3,\ N_\mathrm{s}=1,\ \tilde{N}_\mathrm{b}=2,\ L=4$.
In the case of $N_\mathrm{b}=N_\mathrm{s}$, there is no overlap between blocks (i.e., $\tilde{\mathcal{N}}_l = \emptyset$), which corresponds to the non-overlapping partitioning model used in~\cite{2020Wang_EP_subarray, 2025Arai_XL}.

Based on the partitioned model in the block and overlapping parts, the measurement equations for the $l$-th block part and the $l$-th overlapping part can be expressed as 
\begin{align}
    \label{eq:y_l}
    \mathbf{y}_l = \mathbf{H}_l \mathbf{x} + \mathbf{z}_l, \quad
    \tilde{\mathbf{y}}_l = \tilde{\mathbf{H}}_l \mathbf{x} + \tilde{\mathbf{z}}_l,
\end{align}
where $\mathbf{z}_l \triangleq \mathbf{z}_{\mathcal{N}_l} \in \mathbb{C}^{N_\mathrm{b} \times 1}$ and 
$\tilde{\mathbf{z}}_l \triangleq \mathbf{z}_{\tilde{\mathcal{N}}_l} \in \mathbb{C}^{\tilde{N}_\mathrm{b} \times 1} $ are the noise vectors.
From \eqref{eq:y_l}, the conditional PDFs for the $l$-th block part and the $l$-th overlapping part are given by
\begin{subequations}
\label{eq:p_y_blk}
\begin{align}
    p(\mathbf{y}_l | \mathbf{x}) &= \prod_{n \in \mathcal{N}_l } p(y_{n}| \mathbf{x}) = \mathcal{CN} (\mathbf{y}_l| \mathbf{H}_l \mathbf{x}, \sigma_z \mathbf{I}_{N_\mathrm{b}}), \\
    \label{eq:p_y_blk_tilde}
    p(\tilde{\mathbf{y}}_l | \mathbf{x}) &= \prod_{n \in \tilde{\mathcal{N}}_l } p(y_{n}| \mathbf{x}) = \mathcal{CN} (\tilde{\mathbf{y}}_l| \tilde{\mathbf{H}}_l \mathbf{x}, \sigma_z \mathbf{I}_{\tilde{N}_\mathrm{b}}).
\end{align}
\end{subequations}
For notational convenience, we set $p(\tilde{\mathbf{y}}_L | \mathbf{x})=1$ in \eqref{eq:p_y_blk_tilde}.
Based on the PDFs in \eqref{eq:p_y_blk}, the posterior in \eqref{eq:post} can be rewritten as 
\begin{align}
    \label{eq:post_div}
    p(\mathbf{x}|\mathbf{y}) \propto \prod_{l=1}^{L} \frac{p(\mathbf{y}_{l} | \mathbf{x})}{p(\tilde{\mathbf{y}}_{l} | \mathbf{x})} p(\mathbf{x}).
\end{align}

\subsection{Algorithm Derivation}
\label{subsec:alg_derivation}

To approximately calculate the true posterior $p(\mathbf{x}|\mathbf{y})$ in \eqref{eq:post_div} with reasonable complexity, we design an approximate posterior $g(\mathbf{x}|\mathbf{y})$ in tractable forms.
Based on the EP framework~\cite{2001Minka_EP}, the approximate posterior is designed by minimizing the \ac{KL} divergence, subject to a Gaussian distribution set $\mathbf{\Phi}$, as
\begin{align}
    \label{eq:KL}
    \mathop{\mathrm{minimize}}_{g \in \mathbf{\Phi} } \ 
    \mathrm{KL} \left( p (\mathbf{x}|\mathbf{y}) || g(\mathbf{x}|\mathbf{y}) \right).
\end{align}

The approximate posterior $g(\mathbf{x}|\mathbf{y})$ is designed, using approximate factors $ \{q_l(\mathbf{x}),\ \tilde{q}_l(\mathbf{x}) \}_{l=1}^L$ and $r(\mathbf{x})$, as 
\begin{align}
    g(\mathbf{x}|\mathbf{y}) 
    &= \frac{1}{Z_g} \prod_{l=1}^{L} \frac{q_l(\mathbf{x})}{\tilde{q}_l(\mathbf{x})} r(\mathbf{x}),
\end{align}
where $Z_g$ is a normalizing constant ensuring $\int_\mathbf{x} g(\mathbf{x} | \mathbf{y})=1$.
As can be seen from the correspondence with \eqref{eq:post_div}, these factors serve to approximate the true PDFs as
$q_l(\mathbf{x}) \simeq p(\mathbf{y}_l | \mathbf{x})$,
$\tilde{q}_l(\mathbf{x}) \simeq p(\tilde{\mathbf{y}}_l | \mathbf{x})$,
$r (\mathbf{x}) \simeq  p(\mathbf{x})$.
These approximate factors are further factorized as
\begin{align*}
    q_l(\mathbf{x}) &\triangleq  \prod_{m=1} ^M q_{l,m}(x_m), \quad 
    \tilde{q}_l(\mathbf{x}) \triangleq  \prod_{m=1} ^M \tilde{q}_{l,m}(x_m), \\
    r (\mathbf{x}) &\triangleq  \prod_{m=1} ^M r_{m}(x_m),
\end{align*}
where $q_{l,m}(x_m)$, $\tilde{q}_{l,m}(x_m)$, and $r_{m}(x_m)$ are Gaussian factors parameterized by their means and precisions, defined as
\begin{subequations}
\begin{align}
    q_{l,m} (x_m) &\propto \exp \left ( - \gamma^q_{l,m} |x_m - \mu^q_{l,m}|^2 \right ), \\
    \tilde{q}_{l,m} (x_m) &\propto \exp \left ( - \tilde{\gamma}^q_{l,m} |x_m - \tilde{\mu}^q_{l,m}|^2 \right ), \\
    r_{m} (x_m) &\propto \exp \left ( - \gamma^r_{m} |x_m - \mu^r_{m}|^2 \right ).
\end{align}
\end{subequations}
The parameter vectors are defined as
$\bm{\pi}^q_{l,m} \triangleq [\mu_{l,m}^q, \gamma_{l,m}^q ]^\mathrm{T}$, 
$\tilde{\bm{\pi}}^q_{l,m} \triangleq [\tilde{\mu}_{l,m}^q, \tilde{\gamma}_{l,m}^q ]^\mathrm{T}$, and 
$\bm{\pi}^r_{m} \triangleq [\mu_{m}^r, \gamma_{m}^r ]^\mathrm{T}$,
which represent unknown parameters and are optimized via KL minimization.

Similar to \eqref{eq:post_div}, 
the parameters $\{\tilde{\mu}_{L,m}^q, \tilde{\gamma}_{L,m}^q\}$ corresponding to $l=L$ are fixed as $\tilde{\mu}_{L,m}^q = \tilde{\gamma}_{L,m}^q = 0$ because no overlapping part exists at $l=L$.

Let $\bm{\Pi} \triangleq \{ \bm{\pi}_{l,m}^q, \tilde{\bm{\pi}}_{l,m}^q, \bm{\pi}_{m}^r \}_{\forall l,m}$ denote the set of an unknown parameters.
The optimal parameter set is obtained by minimizing the KL divergence in \eqref{eq:KL}.
However, since the KL divergence involves an intractable integral with respect to the true posterior $p(\mathbf{x})$, a closed-form solution cannot be derived.
To circumvent this issue, we introduce a target distribution $\hat{p}(\mathbf{x}|\mathbf{y})$ in place of the true posterior $p(\mathbf{x}|\mathbf{y})$ in \eqref{eq:post_div}.
This target distribution is formulated by replacing part of the true posterior with approximate factors, as described in the following subsections.

To solve the KL minimization problem, we adopt an alternating optimization approach, where one parameter is updated at a time while keeping the others fixed. 
The methodologies for estimating $\bm{\pi}_{l,m}^q$, $\tilde{\bm{\pi}}_{l,m}^q$ and, $\bm{\pi}_{m}^r$ are detailed in Sections \ref{subsec:update_q}, \ref{subsec:update_q_tilde}, and \ref{subsec:update_r}, respectively.

\subsubsection{Update $\bm{\pi}_{l,m}^q$ for $q_{l,m}(x_m)$}
\label{subsec:update_q}

When updating the parameter $\bm{\pi}_{l,m}^q$, the other parameters $\bm{\Pi} \setminus \bm{\pi}_{l,m}^q$ are fixed as tentative estimates.
Using the cavity function~\cite{2001Minka_EP,2014Cespedes_EP_MIMO}, defined as $g_{\backslash q_l} (\mathbf{x}) \triangleq g(\mathbf{x} | \mathbf{y}) / q_l(\mathbf{x})$, which removes the objective distribution to be updated in the current step $q_l(\mathbf{x})$ from the approximate posterior $g(\mathbf{x}|\mathbf{y})$, the target distribution $\hat{p}_{q_{l}}(\mathbf{x} | \mathbf{y}) $ for updating $\bm{\pi}_{l,m}^q$ is designed as 
\begin{align}
    \label{eq:p_tg_q}
    \hat{p}_{q_{l}}(\mathbf{x} | \mathbf{y}) 
    &\propto p(\mathbf{y}_l | \mathbf{x}) g_{\backslash q_l} (\mathbf{x}).
\end{align}

With the target distribution $\hat{p}_{q_{l}}(\mathbf{x} | \mathbf{y})$ in \eqref{eq:p_tg_q}, 
the minimization of KL divergence, referred to as local divergence~\cite{2005minka_divergence}, is formulated as 
\begin{align}
    \mathop{\mathrm{minimize}}_{\bm{\pi}_{l,m}^q} \ 
    \mathrm{KL} \left( \hat{p}_{q_l} (\mathbf{x}|\mathbf{y}) || g(\mathbf{x}|\mathbf{y}) \right) 
    \triangleq \mathcal{L}_{l,m}^q (\bm{\pi}_{l,m}^q ),
\end{align}
where the objective function $\mathcal{L}_{l,m}^q (\bm{\pi}_{l,m}^q )$ can be expressed as 
\begin{align*}
    \mathcal{L}_{l,m}^q (\bm{\pi}_{l,m}^q )  =  \ln Z_g - \mathbb{E}_{\hat{p}_{q_l}(x_m | \mathbf{y})} \left [ \ln q_{l,m} (x_m) \right ] + \mathrm{const}.
\end{align*}

Since $\mathcal{L}_{l,m}^q$ is convex with respect to $\bm{\pi}_{l,m}^q$, the optimal condition
$\partial \mathcal{L}_{l,m}^q / \partial \bm{\pi}_{l,m}^q = \mathbf{0}$
yields
\begin{align}
    \label{eq:proj_q}
    g(x_m|\mathbf{y}) \propto \mathrm{proj}_{\mathbf{\Phi}} \left [ \hat{p}_{q_{l}}(x_m | \mathbf{y}) \right ],
\end{align}
where $\mathrm{proj}_\mathbf{\Phi}[p(x)] \triangleq \mathcal{CN}(\mathbb{E}_{p(x)}[x], \mathbb{V}_{p(x)} [x] )$ denotes the projection operator onto the Gaussian distribution set $\mathbf{\Phi}$, enforcing the first and second-order moment-matching.

In \eqref{eq:proj_q}, 
the marginal approximate posterior 
$g(x_m|\mathbf{y}) = \int_{\mathbf{x} \backslash x_m} g(\mathbf{x} | \mathbf{y})$ is given by 
\begin{align}
    \label{eq:g_mg_q}
    g(x_m | \mathbf{y}) \propto q_{l,m} (x_m) w_{l,m} (x_m),
\end{align}
where $w_{l,m}(x_m)$ is defined as
\begin{subequations}    
\label{eq:w_update}
\begin{align}
    \label{eq:w_update_1}
    w_{l,m}(x_m) 
    & \triangleq \frac{1}{q_{l,m}(x_{m})} \prod_{l^\prime=1}^L \frac{q_{l^\prime,m}(x_{m})}{\tilde{q}_{l^\prime,m}(x_{m})} r_{m} (x_{m}) \\
    \label{eq:w_update_2}
    & \propto \exp \left ( - \gamma_{l,m}^w |x_m - \mu_{l,m}^w |^2 \right ).
\end{align}
\end{subequations}
The calculation of the precision $\gamma_{l,m}^w$ and the mean $\mu_{l,m}^w$ of $w_{l,m}(x_m)$ will be provided in \eqref{eq:gamma_w}-\eqref{eq:mu_w} in the following subsections.

In \eqref{eq:proj_q}, 
the marginal target distribution $\hat{p}_{q_l} (x_m |\mathbf{y}) = \int_{\mathbf{x} \backslash x_m} \hat{p}_{q_l} (\mathbf{x} | \mathbf{y})$ can be calculated as 
\begin{align}
    \label{eq:p_tg_mg_q}
    \hat{p}_{{q}_{l}}(x_m | \mathbf{y})
    \propto \int_{\mathbf{x} \setminus x_m} 
    \underbrace{p(\mathbf{y}_l | \mathbf{x}) \prod_{m^\prime=1}^M w_{l,m^\prime} (x_{m^\prime})}_{\triangleq u_l(\mathbf{x})}.
\end{align}
Since $p(\mathbf{y}_l|\mathbf{x})$ and $w_{l,m}(x_m)$ in \eqref{eq:p_tg_mg_q} are Gaussian distributions, $u_l(\mathbf{x})$ is also a Gaussian distribution, given by 
\begin{align*}
    u_l(\mathbf{x}) \propto \exp \left [ - (\mathbf{x} - \bm{\mu}_l^u)^\mathrm{H}  (\mathbf{\Sigma}_l^u)^{-1} (\mathbf{x} - \bm{\mu}_l^u) \right ], 
\end{align*}
where the covariance matrix $\mathbf{\Sigma}_{l}^u \in \mathbb{C}^{M \times M}$ and the mean vector $\bm{\mu}_l^u \in \mathbb{C}^{M \times 1}$ can be calculated as 
\begin{subequations}
\label{eq:mu_sigma_u}
\begin{align}
    \mathbf{\Sigma}_{l}^u &= \left ( \sigma_z^{-1} \mathbf{H}_l^\mathrm{H} \mathbf{H}_l + \mathbf{D} \left ( \boldsymbol{\gamma}_{l}^w \right ) \right )^{-1}, \\ 
    \bm{\mu}_l^u &= \mathbf{\Sigma}_{l}^u \left ( \sigma_z^{-1} \mathbf{H}_l^\mathrm{H} \mathbf{y}_l + \boldsymbol{\gamma}_{l}^w \odot \boldsymbol{\mu}_l^w \right ).
\end{align}
\end{subequations}

Using the property of marginal Gaussian distributions~\cite{2006Bishop_PRML}, the marginal target distribution in \eqref{eq:p_tg_mg_q} is expressed as 
\begin{align}
    \label{eq:p_tg_mg_q_2}
    \hat{p}_{{q}_{l}}(x_m | \mathbf{y}) \propto \exp \left ( - \gamma_{l,m}^u |x_m - \mu_{l,m}^u |^2 \right ),
\end{align}
with
$\mu_{l,m}^u = [\bm{\mu}_l^u]_m$ and 
$\gamma_{l,m}^u = 1 / [\mathbf{\Sigma}_l^u]_{m,m}$.

Substituting \eqref{eq:g_mg_q} and \eqref{eq:p_tg_mg_q_2} into \eqref{eq:proj_q}, the approximate factor $q_{l,m} (x_m)$ can be updated as 
\begin{align}
    q_{l,m} (x_m| \mathbf{y}) \propto \hat{p}_{q_l}(x_m| \mathbf{y}) \big / w_{l,m}(x_m), 
\end{align}
where the precision and the mean of $q_{l,m} (x_m)$ are given by 
\begin{subequations}    
\begin{align}
    \gamma_{l,m}^q & = \gamma_{l,m}^u - \gamma_{l,m}^w, \\
    \mu_{l,m}^q &= \left ( \gamma_{l,m}^u  \mu_{l,m}^u - \gamma_{l,m}^w \mu_{l,m}^w \right ) / \gamma_{l,m}^q.
\end{align}
\end{subequations}

\subsubsection{Update $\tilde{\bm{\pi}}_{l,m}^q$ for $\tilde{q}_{l,m}(x_m)$}
\label{subsec:update_q_tilde}
Similar to the update of $\bm{\pi}_{l,m}^q$, the parameter $\tilde{\bm{\pi}}_{l,m}^q$ is updated while the other parameters $\bm{\Pi} \setminus \tilde{\bm{\pi}}_{l,m}^q$ are fixed as tentative estimates.
The minimization of the local divergence for $\tilde{\bm{\pi}}_{l,m}^q$ is then formulated as
\begin{align*}
    \mathop{\mathrm{minimize}}_{\tilde{\bm{\pi}}_{l,m}^q} \ 
    \mathrm{KL} \left( \hat{p}_{\tilde{q}_l} (\mathbf{x}|\mathbf{y}) || g(\mathbf{x}|\mathbf{y}) \right) ,
\end{align*}
where $\hat{p}_{\tilde{q}_l} (\mathbf{x}|\mathbf{y})$ is the target distribution, given by
\begin{align}
    \label{eq:p_tg_q_tilde}
    \hat{p}_{\tilde{q}_l} (\mathbf{x}|\mathbf{y})
    &\propto p(\tilde{\mathbf{y}}_l | \mathbf{x}) g_{\backslash \tilde{q}_l} (\mathbf{x}),
\end{align}
with the cavity distribution 
$g_{\backslash \tilde{q}_l} (\mathbf{x}) \triangleq g(\mathbf{x} |\mathbf{y}) / \tilde{q}_{l,m} (x_m)$.

Through the same procedure as that used for  deriving $\bm{\pi}_{l,m}^q$, the optimal condition of $\tilde{\bm{\pi}}_{l,m}^q$ is given by 
\begin{align}
    \label{eq:proj_q_tilde}
    g(x_m|\mathbf{y}) \propto \mathrm{proj}_{\mathbf{\Phi}} \left [ \hat{p}_{\tilde{q}_{l}}(x_m | \mathbf{y}) \right ].
\end{align}

The marginal approximate posterior $g(x_m|\mathbf{y}) = \int_{\mathbf{x} \backslash x_m} g(\mathbf{x} | \mathbf{y})$ in \eqref{eq:proj_q_tilde} is expressed as
\begin{align}
    \label{eq:g_mg_q_tilde}
    g(x_m | \mathbf{y}) \propto \tilde{q}_{l,m} (x_m) \tilde{w}_{l,m} (x_m),
\end{align}
where $\tilde{w}_{l,m}(x_m)$ is defined as
\begin{subequations}    
\label{eq:w_tilde_update}
\begin{align}
    \tilde{w}_{l,m}(x_m) 
    & \triangleq \frac{1}{\tilde{q}_{l,m}(x_{m})} \prod_{l^\prime=1}^L \frac{q_{l^\prime,m}(x_{m})}{\tilde{q}_{l^\prime,m}(x_{m})} r_{m} (x_{m}) \\
    & \propto \exp \left ( - \tilde{\gamma}_{l,m}^w |x_m - \tilde{\mu}_{l,m}^w |^2 \right ).
\end{align}
\end{subequations}

The calculation of the precision $\tilde{\gamma}_{l,m}^w$ and the mean $\tilde{\mu}_{l,m}^w$ of $\tilde{w}_{l,m}(x_m)$ will be provided in \eqref{eq:gamma_w_tilde}-\eqref{eq:mu_w_tilde} in the following subsections.

In \eqref{eq:proj_q_tilde}, 
the marginal target distribution $\hat{p}_{\tilde{q}_l} (x_m |\mathbf{y}) = \int_{\mathbf{x} \backslash x_m} \hat{p}_{\tilde{q}_l} (\mathbf{x} | \mathbf{y})$ is expressed as 
\begin{align}
    \label{eq:p_tg_mg_q_tilde}
    \hat{p}_{\tilde{q}_{l}}(x_m | \mathbf{y})
    \propto \int_{\mathbf{x} \setminus x_m} 
    \underbrace{p(\tilde{\mathbf{y}}_l | \mathbf{x}) \prod_{m^\prime=1}^M \tilde{w}_{l,m^\prime} (x_{m^\prime})}_{\triangleq \tilde{u}_l(\mathbf{x})},
\end{align}
where $\tilde{u}_l(\mathbf{x})$ can be expressed as 
\begin{align*}
    \tilde{u}_l(\mathbf{x}) 
    \propto \exp \left [ - (\mathbf{x} - \tilde{\bm{\mu}}_l^u)^\mathrm{H}  (\tilde{\mathbf{\Sigma}}_l^u)^{-1} (\mathbf{x} - \tilde{\bm{\mu}}_l^u) \right ],
\end{align*}
and the covariance matrix $\tilde{\mathbf{\Sigma}}_{l}^u \in \mathbb{C}^{M \times M}$ and the mean vector $\tilde{\bm{\mu}}_l^u \in \mathbb{C}^{M \times 1}$ are given by 
\begin{subequations}
\label{eq:mu_sigma_u_tilde}
\begin{align}
    \tilde{\mathbf{\Sigma}}_{l}^u &= \left ( \sigma_z^{-1} \tilde{\mathbf{H}}_l^\mathrm{H} \tilde{\mathbf{H}}_l + \mathbf{D} \left ( \tilde{\boldsymbol{\gamma}}_{l}^w \right ) \right )^{-1}, \\
    \tilde{\bm{\mu}}_l^u &= \tilde{\mathbf{\Sigma}}_{l}^u \left ( \sigma_z^{-1} \tilde{\mathbf{H}}_l^\mathrm{H} \tilde{\mathbf{y}}_l + \tilde{\boldsymbol{\gamma}}_{l}^w \odot \tilde{\boldsymbol{\mu}}_l^w \right ).
\end{align}
\end{subequations}

Using the property of marginal Gaussian distributions~\cite{2006Bishop_PRML}, the marginal target distribution in \eqref{eq:p_tg_mg_q_tilde} is expressed as 
\begin{align}
    \label{eq:p_tg_mg_q_tilde_2}
    \hat{p}_{\tilde{q}_{l}}(x_m | \mathbf{y}) \propto \exp \left ( - \tilde{\gamma}_{l,m}^u |x_m - \tilde{\mu}_{l,m}^u |^2 \right ),
\end{align}
with
$\tilde{\mu}_{l,m}^u = [\tilde{\bm{\mu}}_l^u]_m$ and 
$\tilde{\gamma}_{l,m}^u = 1 / [\tilde{\mathbf{\Sigma}}_l^u]_{m,m}$.

Substituting \eqref{eq:g_mg_q_tilde} and \eqref{eq:p_tg_mg_q_tilde_2} into \eqref{eq:proj_q_tilde}, the approximate factor $\tilde{q}_{l,m} (x_m)$ can be updated as 
\begin{align*}
    \tilde{q}_{l,m} (x_m| \mathbf{y}) \propto \hat{p}_{\tilde{q}_l}(x_m| \mathbf{y}) \big / \tilde{w}_{l,m}(x_m) , 
\end{align*}
where the precision and the mean of $\tilde{q}_{l,m} (x_m)$ are given by 
\begin{subequations}    
\begin{align}
    \tilde{\gamma}_{l,m}^q & = \tilde{\gamma}_{l,m}^u - \tilde{\gamma}_{l,m}^w, \\
    \tilde{\mu}_{l,m}^q &= \left ( \tilde{\gamma}_{l,m}^u  \tilde{\mu}_{l,m}^u - \tilde{\gamma}_{l,m}^w \tilde{\mu}_{l,m}^w \right ) / \tilde{\gamma}_{l,m}^q.
\end{align}
\end{subequations}

\subsubsection{Update $\bm{\pi}_{m}^r$ for $r_m(x_m)$}
\label{subsec:update_r}
The parameter $\tilde{\bm{\pi}}_{m}^r$ is updated while the other parameters $\bm{\Pi} \setminus \tilde{\bm{\pi}}_{m}^r$ are fixed as tentative estimates.
The minimization of the local divergence for $\tilde{\bm{\pi}}_{m}^r$ is then formulated as
\begin{align*}
    \mathop{\mathrm{minimize}}_{\tilde{\bm{\pi}}_{m}^r} \ 
    \mathrm{KL} \left( \hat{p}_{r_m} (\mathbf{x}|\mathbf{y}) || g(\mathbf{x}|\mathbf{y}) \right) ,
\end{align*}
where $\hat{p}_{r_m} (\mathbf{x}|\mathbf{y})$ is the target distribution, given by
\begin{align}
    \label{eq:p_tg_r}
    \hat{p}_{r_m} (\mathbf{x}|\mathbf{y})
    &\propto p(\mathbf{x}) g_{\backslash r_m} (\mathbf{x}),
\end{align}
with the cavity distribution 
$g_{\backslash r_m} (\mathbf{x}) \triangleq g(\mathbf{x} |\mathbf{y}) / r_m (x_m)$.

Through the same procedure as that used for deriving $\bm{\pi}_{l,m}^q$, the optimal condition of $\tilde{\bm{\pi}}_{m}^r$ is given by 
\begin{align}
    \label{eq:proj_r}
    g(x_m|\mathbf{y}) \propto \mathrm{proj}_{\mathbf{\Phi}} \left [ \hat{p}_{r_m}(x_m | \mathbf{y}) \right ],
\end{align}
where the marginal target distribution $\hat{p}_{r_{m}}(x_m | \mathbf{y})$ in \eqref{eq:proj_r} can be expressed as
\begin{align}
    \label{eq:p_tg_r_2}
    \hat{p}_{r_{m}}(x_m | \mathbf{y}) &\propto p(x_m) \underbrace{\prod_{l^\prime = 1}^L \frac{q_{l^\prime, m}(x_m)}{\tilde{q}_{l^\prime, m}(x_m)}}_{\triangleq \bar{q}_{m}(x_m) }.
\end{align}

Since $q_{l,m}(x_m)$ and $\tilde{q}_{l,m}(x_m)$ are Gaussian distributions, their product $\bar{q}_{m}(x_m)$ is also a Gaussian distribution, which can be expressed as 
\begin{align*}
    \bar{q}_m(x_m) \propto \exp \left ( - \bar{\gamma}_m^q | x_m - \bar{\mu}_m^q |^2 \right ),
\end{align*}
where the precision $\bar{\gamma}^q_m$ and the mean $\bar{\mu}^q_m$ are given by 
\begin{subequations}
\label{eq:combine}
\begin{align}
    \bar{\gamma}_m^q &= \sum_{l=1}^L \left ( \gamma_{l,m}^q - \tilde{\gamma}_{l,m}^q \right ), \\
    \bar{\mu}_m^q &= \frac{1}{\bar{\gamma}_m^q}
    \sum_{l=1}^L \left ( \gamma_{l,m}^q \mu_{l,m}^q - \tilde{\gamma}_{l,m}^q \tilde{\mu}_{l,m}^q \right ).
\end{align}
\end{subequations}

Substituting \eqref{eq:p_tg_r_2} into \eqref{eq:proj_r}, the marginal approximate posterior $g(x_m|\mathbf{y})$ can be expressed as
\begin{align}
    \label{eq:g_mg}
    g(x_m|\mathbf{y}) \propto \mathrm{proj}_\mathbf{\Phi} \left [ p(x_m) \bar{q}_m (x_m) \right ] 
    \!=\! \mathcal{CN} (x_m | \hat{x}_m, \hat{v}_m).
\end{align}

Using the MMSE denoiser function~\cite{2019Rangan_VAMP}
$\eta_\mathrm{e}(\bar{\mu}^q_m, \bar{\gamma}^q_m) \triangleq \mathbb{E}_{p(x_m) \bar{q}_m(x_m)} [x_m]$,
the approximate posterior mean and variance are obtained as 
$\hat{x}_m = \eta_\mathrm{e}(\bar{\mu}^q_m, \bar{\gamma}^q_m)$ and 
$\hat{v}_m = \eta_\mathrm{v} (\bar{\mu}^q_m, \bar{\gamma}^q_m) \triangleq \bar{\gamma}^q_m \frac{\partial \eta_\mathrm{e}(\bar{\mu}^q_m, \bar{\gamma}^q_m) }{\partial \bar{\mu}^q_m}$.
Substituting the prior $p(x_m)$ based on \ac{QAM} constellation into \eqref{eq:g_mg}, the approximate posterior mean $\hat{x}_m$ and variance $\hat{v}_m$ can be explicitly calculated as
\begin{subequations}
\label{eq:mu_sigma_denoiser}
\begin{align}
    \hat{x}_m &= \frac{1}{Z_m} \sum_{\mathrm{x} \in \mathcal{X}} \mathrm{x} \exp \left ( - \bar{\gamma}_{m}^q |\mathrm{x} - \bar{\mu}^q_m|^2 \right ), \\
    \hat{v}_m &= \frac{1}{Z_m} \sum_{\mathrm{x} \in \mathcal{X}} |\mathrm{x}|^2 \exp \left ( - \bar{\gamma}_{m}^q |\mathrm{x} - \bar{\mu}^q_m|^2 \right ) - |\hat{x}_m|^2, 
\end{align}
\end{subequations}
where $Z_m$ is the normalizing constant given by 
$Z_m = \sum_{\mathrm{x} \in \mathcal{X}} \exp \left ( - \bar{\gamma}_{m}^q |\mathrm{x} - \bar{x}^q_m|^2 \right )$.
For notational convenience, the precision of the denoiser output is defined as $\hat{\gamma}_m \triangleq \hat{v}_m^{-1}$.

Since the marginal approximate posterior $g(x_m | \mathbf{y})$ is expressed as 
\begin{align*}
    g(x_m | \mathbf{y}) &\propto r_m(x_m) \bar{q}_m(x_m), 
\end{align*}
the approximate factor $r_m(x_m)$ can be updated, from \eqref{eq:g_mg}, as
%
\begin{align}
    \label{eq:r_update}
    r_{m} (x_m) \propto \mathcal{CN}(x_m| \hat{x}_m, \hat{\gamma}_m^{-1}) \big / \bar{q}_m(x_m),
\end{align}
where 
the precision and the mean of $r_{m} (x_m)$ are given by 
\begin{subequations}    
\begin{align}
    \gamma_m^r &= \hat{\gamma}_m - \bar{\gamma}_m^q, \\
    \mu_m^r &= \left ( \hat{\gamma}_m \hat{x}_m - \bar{\gamma}_m^q \bar{\mu}_m^q \right ) / \gamma_m^r .
\end{align}
\end{subequations}
 
Substituting \eqref{eq:r_update} into \eqref{eq:w_update}, $w_{l,m}(x_m)$ can be updated as
\begin{align*}
    w_{l,m}(x_m) \propto \mathcal{CN}(x_m| \hat{x}_m, \hat{\gamma}_m^{-1}) \big / q_{l,m}(x_m),
\end{align*}
where the precision and the mean of $w_{l,m}(x_m)$ are given by 
\begin{subequations} 
\label{eq:mu_gamma_w}
\begin{align}
    \label{eq:gamma_w}
    \gamma_{l,m}^w &= \hat{\gamma}_m - \gamma^q_{l,m}, \\
    \label{eq:mu_w}
    \mu_{l,m}^w &= \left ( \hat{\gamma}_m \hat{x}_m - \gamma_{l,m}^q \mu_{l,m}^q \right ) / \gamma_{l,m}^w.
\end{align}
\end{subequations}

Similarly, substituting \eqref{eq:r_update} into \eqref{eq:w_tilde_update}, $\tilde{w}_{l,m} (x_m)$ can be updated as 
\begin{align*}
    \tilde{w}_{l,m}(x_m) \propto \mathcal{CN}(x_m| \hat{x}_m, \hat{\gamma}_m^{-1}) \big / \tilde{q}_{l,m}(x_m),
\end{align*}
where the precision and the mean of $\tilde{w}_{l,m}(x_m)$ are given by
\begin{subequations}   
\label{eq:mu_gamma_w_tilde} 
\begin{align}
    \label{eq:gamma_w_tilde}
    \tilde{\gamma}_{l,m}^w &= \hat{\gamma}_m - \tilde{\gamma}^q_{l,m}, \\
    \label{eq:mu_w_tilde}
    \tilde{\mu}_{l,m}^w &= \left ( \hat{\gamma}_m \hat{x}_m - \tilde{\gamma}_{l,m}^q \tilde{\mu}_{l,m}^q \right ) / \tilde{\gamma}_{l,m}^w.
\end{align}
\end{subequations}

\subsection{Algorithm Description}
\label{subsec:Algorithm_Description}

The pseudocode of the proposed algorithm, \textit{OvEP}, is summarized in Algorithm~\ref{alg:OvEP}.
In the proposed algorithm, LMMSE-based filters are designed for each block part and overlapping part in \eqref{eq:mu_sigma_u} and \eqref{eq:mu_sigma_u_tilde} (lines 3–12 of Algorithm~\ref{alg:OvEP}).
Subsequently, the LMMSE outputs from each block part $\bm{\mu}_l^q$ and overlapping part $\tilde{\bm{\mu}}_l^q$ are combined through weighted averaging based on their precisions $\bm{\gamma}_l^q$ and $\tilde{\bm{\gamma}}_l^q$ in \eqref{eq:combine} (lines 13-14).
To account for the overlapping block structure, the overlapping parts are subtracted during the combining process, yielding the denoiser input $\{\bar{\bm{\mu}}^q, \bar{\bm{\gamma}}^q \}$.
By employing the MMSE denoiser functions 
$\eta_\mathrm{e} (\cdot)$ and 
$\eta_\mathrm{v} (\cdot)$, designed based on the prior $p(\mathbf{x})$, the approximate posterior mean $\hat{\mathbf{x}}$ and variance $\hat{\mathbf{v}}$ are computed from the denoiser input $\{\bar{\bm{\mu}}^q, \bar{\bm{\gamma}}^q \}$.
Finally, using the denoiser output $\{ \hat{\mathbf{x}}, \hat{\mathbf{v}} \}$, the extrinsic values for block parts $\{ \bm{\mu}_l^w, \bm{\gamma}_l^w \}$ and overlapping parts $\{ \tilde{\bm{\mu}}_l^w, \tilde{\bm{\gamma}}_l^w \}$ are computed based on moment-matching rules in \eqref{eq:mu_gamma_w} and \eqref{eq:mu_gamma_w_tilde} (lines 18-21).
This process is repeated until the number of iterations reaches the predefined maximum number $T$.

The proposed method represents a generalized framework that encompasses conventional algorithms as special cases.
By setting the parameters $N_\mathrm{b}$, $N_\mathrm{s}$, and $L$ to specific values, the proposed algorithm reduces to the following well-known algorithms:
\begin{itemize}
    \item \textbf{\textit{LMMSE-EP}}~\cite{2014Cespedes_EP_MIMO} with $N_\mathrm{b}=N_\mathrm{s}=N,\ L=1$:
    This corresponds to the EP algorithm without block partitioning.
    Since full-dimensional LMMSE-based filtering with matrix inversion is required at each iteration, the algorithm incurs the highest computational complexity.
    \item \textbf{\textit{MF-EP}}~\cite{2015Meng_MFEP} with $N_\mathrm{b}=N_\mathrm{s}=1,\ L=N$:
    This corresponds to the EP algorithm that partitions the measurements into scalar values.
    As it does not require matrix inversion, this algorithm has the lowest computational complexity.
    \item \textbf{\textit{NOvEP}}~\cite{2020Wang_EP_subarray} with $N_\mathrm{b}=N_\mathrm{s} \leq N,\ L=N/N_\mathrm{b}$:
    This corresponds to the EP algorithm that partitions the measurements into non-overlapping block parts.
\end{itemize}

In EP-based algorithms, the update of the precision parameters $\bar{\bm{\gamma}}^q$, $\bm{\gamma}_l^w$, and $\tilde{\bm{\gamma}}_l^w$ in lines 13, 18, and 20 of Algorithm~\ref{alg:OvEP} may yield negative values.
Such cases indicate that an optimal parameter pair of mean and precision does not exist within the range of positive precisions.
Following~\cite{2001Minka_EP_thesis,2014Cespedes_EP_MIMO}, the corresponding mean and precision parameters are then left unchanged, retaining their previous values.

\begin{algorithm}[t]
    \caption[]{Overlapping block partitioning EP (OvEP)}
    \label{alg:OvEP}
    \hrulefill
    \begin{algorithmic}[1]
        \vspace{-0.5ex}
        \Statex \textbf{Input:} 
            $\{\mathbf{y}_l, \mathbf{H}_l \}_{l=1}^L,\ \{ \tilde{\mathbf{y}}_l, \tilde{\mathbf{H}}_l \}_{l=1}^{L-1}$
            $\sigma_x,\ \sigma_z,\ \beta_x$
        \Statex \textbf{Output:} 
            $\hat{\mathbf{x}}, \hat{\mathbf{v}}$
        \vspace{-1.5ex}
        \Statex \hspace{-3ex} \hrulefill

        \Statex \textbf{Initialization:} 
        \State 
            $\boldsymbol{\gamma}_l^w = \tilde{\boldsymbol{\gamma}}_l^w = \beta_x^{-1} \mathbf{1}_{M},\ \boldsymbol{\mu}_l^w = \tilde{\boldsymbol{\mu}}_l^w = \mathbf{0}_M$

        \For{$t=1, 2, \ldots, T$} 

        \Statex 
            \textbf{\quad // LMMSE filter in the block parts} 
        \State 
            $\mathbf{\Sigma}_{l}^u = \left ( \sigma_z^{-1} \mathbf{H}_l^\mathrm{H} \mathbf{H}_l + \mathbf{D} \left ( \boldsymbol{\gamma}_{l}^w \right)\right )^{-1}$
        \State 
            $\boldsymbol{\mu}_l^u = \mathbf{\Sigma}_{l}^u \left ( \sigma_z^{-1} \mathbf{H}_l^\mathrm{H} \mathbf{y}_l + \boldsymbol{\gamma}_{l}^w \odot \boldsymbol{\mu}_l^w \right )$
        \State 
            $\boldsymbol{\gamma}_l^u = \mathbf{1}_M \oslash \mathbf{d} (\mathbf{\Sigma}_l^u)$
        \State
            $\boldsymbol{\gamma}_l^q = \boldsymbol{\gamma}_l^u - \boldsymbol{\gamma}_l^w$
        \State
            $\boldsymbol{\mu}_l^q = \left ( \boldsymbol{\gamma}_l^u \odot \boldsymbol{\mu}_l^u - \boldsymbol{\gamma}_l^w \odot \boldsymbol{\mu}_l^w \right ) \oslash \boldsymbol{\gamma}_l^q$

        \Statex 
            \textbf{\quad // LMMSE filter in the overlapping parts}
        \State 
            $\tilde{\mathbf{\Sigma}}_{l}^u = \left ( \sigma_z^{-1} \tilde{\mathbf{H}}_l^\mathrm{H} \tilde{\mathbf{H}}_l + \mathbf{D} \left ( \tilde{\boldsymbol{\gamma}}_{l}^w \right) \right )^{-1}$
        \State 
            $\tilde{\boldsymbol{\mu}}_l^u = \tilde{\mathbf{\Sigma}}_{l}^u \left ( \sigma_z^{-1} \tilde{\mathbf{H}}_l^\mathrm{H} \tilde{\mathbf{y}}_l + \tilde{\boldsymbol{\gamma}}_{l}^w \odot \tilde{\boldsymbol{\mu}}_l^w \right )$
        \State 
            $\tilde{\boldsymbol{\gamma}}_l^u = \mathbf{1}_M \oslash \mathbf{d} (\tilde{\mathbf{\Sigma}}_l^u)$
        \State
            $\tilde{\boldsymbol{\gamma}}_l^q = \tilde{\boldsymbol{\gamma}}_l^u - \tilde{\boldsymbol{\gamma}}_l^w$, \hfill
            $(\tilde{\bm{\mu}}_L^q = \mathbf{0}_M)$
        \State
            $\tilde{\boldsymbol{\mu}}_l^q = \left ( \tilde{\boldsymbol{\gamma}}_l^u \odot \tilde{\boldsymbol{\mu}}_l^u - \tilde{\boldsymbol{\gamma}}_l^w \odot \tilde{\boldsymbol{\mu}}_l^w \right ) \oslash \tilde{\boldsymbol{\gamma}}_l^q$, \hfill
            $(\tilde{\bm{\gamma}}_L^q = \mathbf{0}_M)$

        \Statex 
            \textbf{\quad // Combining the block and overlapping parts}
        \State 
            $\bar{\boldsymbol{\gamma}}^q = \sum_{l=1}^L \left ( \boldsymbol{\gamma}_l^q - \tilde{\boldsymbol{\gamma}}_l^q \right )$
        \State 
            $\bar{\boldsymbol{\mu}}^q = \sum_{l=1}^L \left ( \boldsymbol{\gamma}_l^q \odot \boldsymbol{\mu}_l^q - \tilde{\boldsymbol{\gamma}}_l^q \odot \tilde{\boldsymbol{\mu}}_l^q \right ) \oslash \bar{\boldsymbol{\gamma}}^q$

        \Statex 
            \textbf{\quad // MMSE denoiser based on prior}
        \State 
            $\mathbf{\hat{x}} = \eta_\mathrm{e} (\bar{\boldsymbol{\mu}}^q, \bar{\boldsymbol{\gamma}}^q)$
        \State 
            $ \hat{\mathbf{v}} =  \eta_\mathrm{v} (\bar{\boldsymbol{\mu}}^q, \bar{\boldsymbol{\gamma}}^q)$
        \State
            $\hat{\boldsymbol{\gamma}} = \mathbf{1}_M \oslash \hat{\mathbf{v}}$

        \Statex 
            \textbf{\quad // Extrinsic value generation for block parts}
        \State 
            $\boldsymbol{\gamma}_l^w = \hat{\boldsymbol{\gamma}} - \boldsymbol{\gamma}^q_l$
        \State 
            $\boldsymbol{\mu}_l^w = \left ( \hat{\boldsymbol{\gamma}} \odot \hat{\mathbf{x}} - \boldsymbol{\gamma}_l^q \odot \boldsymbol{\mu}_l^q \right ) \oslash \boldsymbol{\gamma}_l^w$

        \Statex 
            \textbf{\quad // Extrinsic value generation for overlapping parts}
        \State 
            $\tilde{\boldsymbol{\gamma}}_l^w = \hat{\boldsymbol{\gamma}} - \tilde{\boldsymbol{\gamma}}^q_l$
        \State 
            $\tilde{\boldsymbol{\mu}}_l^w = \left ( \hat{\boldsymbol{\gamma}} \odot \hat{\mathbf{x}} - \tilde{\boldsymbol{\gamma}}_l^q \odot \tilde{\boldsymbol{\mu}}_l^q \right ) \oslash \tilde{\boldsymbol{\gamma}}_l^w$
        \EndFor
    \end{algorithmic}
\end{algorithm}

\subsection{Analysis of the Denoiser Input}
\label{subsec:denoiser_in}

One of the essential steps in the proposed algorithm is the subtraction of the overlapping part when calculating the denoiser input $\{ \bar{\bm{\mu}}^q, \bar{\bm{\gamma}}^q\}$ in lines 13-14 of Algorithm~\ref{alg:OvEP}.
This subsection therefore analyzes the denoiser input of the proposed algorithm.

For notational convenience, we define the LMMSE output obtained after the subtraction of the $l$-th overlapping part from the $l$-th block part as 
\begin{align}
    \label{eq:mu_bar_q}
    \bar{\bm{\mu}}_l^q \triangleq 
    \underbrace{(\bm{\gamma}_l^q \oslash \bar{\bm{\gamma}}^q )\odot \bm{\mu}_l^q}_{\text{Block part}} 
    - \underbrace{(\tilde{\bm{\gamma}}_l^q \oslash \bar{\bm{\gamma}}^q ) \odot \tilde{\bm{\mu}}_l^q}_{\text{Overlapping part}},
\end{align}
with $\tilde{\bm{\mu}}_L^q = \tilde{\bm{\gamma}}_L^q = \mathbf{0}_M$.

From \eqref{eq:mu_bar_q}, the denoiser input in line 14 of Algorithm~\ref{alg:OvEP} is expressed as 
$\bar{\bm{\mu}}^q = \sum_{l=1}^L \bar{\bm{\mu}}_l^q$.
Based on this expression, the MSE of the denoiser input can be expressed as
\begin{align}
    \label{eq:MSE_in}
    &\mathrm{MSE} (\bar{\bm{\mu}}^q)
    \triangleq \mathbb{E}_{p(\mathbf{H}, \mathbf{x}, \mathbf{z})} \left [ \| \bar{\boldsymbol{\mu}}^q - \mathbf{x} \|_2^2 \right ] \nonumber \\
    &= \sum_{l=1}^L \sum_{l^\prime=1}^L \mathbb{E} \left [ 
    \bar{\boldsymbol{\mu}}^{q\mathrm{H}}_l
    \bar{\boldsymbol{\mu}}^{q}_{l^\prime}  \right ]
    -2 \mathfrak{R} \left \{ 
    \sum_{l=1}^L \mathbb{E} \left [ \mathbf{x}^\mathrm{H} \bar{\boldsymbol{\mu}}_l^q \right ] \right \}
    + \mathbb{E} \left [ \| \mathbf{x} \|_2^2 \right ],
\end{align}
where the first term $\sum_{l,l^\prime} \mathbb{E} [ \bar{\boldsymbol{\mu}}^{q\mathrm{H}}_l \bar{\boldsymbol{\mu}}^{q}_{l^\prime}  ]$ represents the sum of the inner products between the LMMSE outputs across all block combinations, which quantifies the inter-block correlation among the LMMSE outputs.
The second term $\sum_{l} \mathbb{E} \left [ \mathbf{x}^\mathrm{H} \bar{\boldsymbol{\mu}}_l^q \right ] $ represents the sum of the inner products between the LMMSE outputs $\bar{\bm{\mu}}_l^q$ and the target vector $\mathbf{x}$, which quantifies the alignment of $\bar{\bm{\mu}}_l^q$ to $\mathbf{x}$.
The third term is the constant term calculated as $\mathbb{E} \left [ \| \mathbf{x} \|_2^2 \right ] = M \sigma_x$.
Hence, the MSE can be reduced by employing an estimator with low inter-block correlation and high alignment with the target vector.
Regarding the MSE of the denoiser input in the proposed algorithm, the following Proposition holds.
\begin{proposition}
    \label{proposition:mse}
    The MSE of the denoiser input defined in \eqref{eq:MSE_in} at the initial iteration $(t=1)$ of Algorithm~\ref{alg:OvEP} can be approximated as
    \begin{align}
        \label{eq:MSE_in_apx}
        \mathrm{MSE} (\bar{\bm{\mu}}^q) \simeq 
        \sum_{l=1}^L \sum_{l^\prime=1}^L \mathbb{E}_{p(\mathbf{H}, \mathbf{x}, \mathbf{z})} \left [ 
        \bar{\bm{\mu}}^{q\mathrm{H}}_l
        \bar{\bm{\mu}}^{q}_{l^\prime}  \right ] - \sigma_x M ,
    \end{align}
    provided that the conditions $\bm{\gamma}_l^w \ll \bm{\gamma}_l^u$ and $\tilde{\bm{\gamma}}_l^w \ll \tilde{\bm{\gamma}}_l^u$ hold.
\end{proposition}
\begin{proof}
    See Appendix~\ref{apx:mse}.
\end{proof}

From Proposition~\ref{proposition:mse}, the MSE of the denoiser input depends only on the inter-block correlation among the LMMSE outputs since $\sigma_x M$ is constant in \eqref{eq:MSE_in_apx}.
Thus, reducing the inter-block correlation is crucial for improving the MSE performance.
When the channel matrix exhibits strong spatial correlation, the LMMSE outputs across blocks also become highly correlated.
Consequently, the conventional \textit{NOvEP} suffers from high inter-block correlation, which degrades the MSE performance.
In contrast, the proposed method partitions the large-scale measurements into overlapping blocks and subtracts the LMMSE outputs in the overlapping parts from those of the corresponding block parts, as defined in \eqref{eq:mu_bar_q}.
This subtraction operation effectively mitigates inter-block correlation.

Furthermore, reducing the inter-block correlation among the LMMSE outputs helps suppress outliers in the denoiser input.
If such outliers occur at the early stages of the iteration process, the denoiser function $\eta_\mathrm{e}(\cdot)$ operates in the \textit{inactive} region~\cite{2024Furudoi_ADD} (also referred to as the hard-valued region~\cite{2019Takahashi_ASB}), which produces hard-decision outputs, resulting in error propagation across iterations.
Therefore, reducing the MSE of the denoiser input at the initial iteration is particularly important for preventing outliers in \acp{MPA}~\cite{2022Tamaki_GAMP_feedback,2024Furudoi_ADD,2019Takahashi_ASB}.
A detailed evaluation for the denoiser input will be provided in Section~\ref{sec:simulation}.

At the initial iteration $(t=1)$, to reflect weak prior knowledge, the prior variance $\beta_x$ should be set to a large value, thereby relying more heavily on the measurements rather than the prior information.
Under this setting, the assumptions $\bm{\gamma}_l^w \ll \bm{\gamma}_l^u$ and $\tilde{\bm{\gamma}}_l^w \ll \tilde{\bm{\gamma}}_l^u$ hold since $\bm{\gamma}_l^w = \tilde{\bm{\gamma}}_l^w = \beta_x^{-1} \mathbf{1}_M$ at $t=1$.

\subsection{Analysis of Fixed Points}

As described in Section~\ref{subsec:alg_derivation}, the proposed algorithm iteratively solves decomposed local divergence minimization via alternating optimization, instead of directly solving the original KL divergence minimization.
To validate the effectiveness of this approach, this section demonstrates that the fixed point of the proposed algorithm corresponds to a stationary point of a relaxed version of the original KL divergence.

We begin by characterizing the fixed point of Algorithm~\ref{alg:OvEP} as follows:
\begin{lemma}
    Any fixed point of Algorithm~\ref{alg:OvEP} satisfies 
    \begin{subequations}
    \label{eq:fixed_point}
    \begin{align}
        \boldsymbol{\gamma}_1^u &= \cdots = \boldsymbol{\gamma}_L^u = \tilde{\boldsymbol{\gamma}}_1^u = \cdots = \tilde{\boldsymbol{\gamma}}_{L-1}^u = \boldsymbol{\hat{\gamma}} = \boldsymbol{\gamma}_\mathrm{f},  \\
        \boldsymbol{\mu}_1^u &= \cdots = \boldsymbol{\mu}_L^u = \tilde{\boldsymbol{\mu}}_1^u = \cdots = \tilde{\boldsymbol{\mu}}_{L-1}^u = \hat{\mathbf{x}} = \mathbf{x}_\mathrm{f},
    \end{align}
    \end{subequations}
    where $\bm{\gamma}_\mathrm{f}$ and $\mathbf{x}_\mathrm{f}$ are defined as
    \begin{subequations}
    \label{eq:fixed_point_2}
    \begin{align}
        \boldsymbol{\gamma}_\mathrm{f} 
        &\triangleq \bar{\boldsymbol{\gamma}}^q + \sum_{l=1}^L \left ( \boldsymbol{\gamma}_l^w - \tilde{\boldsymbol{\gamma}}_l^w \right ), \\
        \mathbf{x}_\mathrm{f} 
        & \! \triangleq \! \left (  \bar{\boldsymbol{\gamma}}^q \odot \bar{\boldsymbol{\mu}}^q + \sum_{l=1}^L \left ( \boldsymbol{\gamma}_l^w \odot \boldsymbol{\mu}_l^w
        - \tilde{\boldsymbol{\gamma}}_l^w \odot \tilde{\boldsymbol{\mu}}_l^w \right )
        \right ) \oslash \boldsymbol{\hat{\gamma}}.
    \end{align}
    \end{subequations}
\end{lemma}

\begin{proof}
    Substituting lines 18-19 and 20-21 of Algorithm~\ref{alg:OvEP} into lines 6-7 and 11-12 yields 
    $\bm{\gamma}_l^u = \hat{\bm{\gamma}},\ \bm{\mu}_l^u = \hat{\mathbf{x}}$ and 
    $\tilde{\bm{\gamma}}_l^u = \hat{\bm{\gamma}},\ \tilde{\bm{\mu}}_l^u = \hat{\mathbf{x}}$.
    Then, substituting lines 6, 11 into 13 yields
    $\hat{\bm{\gamma}} = \bm{\gamma}_\mathrm{f}$, while 
    substituting lines 7, 12 into 14 yields $\hat{\mathbf{x}} = \mathbf{x}_\mathrm{f}$.
\end{proof}

This lemma indicates that the outputs of the denoiser $\{\hat{\bm{\gamma}}, \hat{\mathbf{x}} \}$ and those of the LMMSE filters in both the block part $\{\bm{\gamma}_l^u, \bm{\mu}_l^u \}$ and the overlapping part $\{\tilde{\bm{\gamma}}_l^u, \tilde{\bm{\mu}}_l^u \}$ are identical at any fixed point.
In what follows, we show that this fixed point in \eqref{eq:fixed_point_2} corresponds to a stationary point of a relaxed version of the original KL divergence defined in \eqref{eq:KL}.

For notational simplicity, we define the PDFs as 
$f_0(\mathbf{x}) \triangleq p(\mathbf{x})$, 
$f_l(\mathbf{x}) \triangleq p(\mathbf{y}_l | \mathbf{x}),\ \forall l \in \{1,\ldots, L\}$, and
$\tilde{f}_l(\mathbf{x}) \triangleq p(\tilde{\mathbf{y}}_l | \mathbf{x}),\ \forall l \in \{1,\ldots, L-1\}$.
By substituting \eqref{eq:post_div} into \eqref{eq:KL}, the original KL divergence between the true posterior $p(\mathbf{x}|\mathbf{y})$ and the approximate posterior $g(\mathbf{x}|\mathbf{y})$ can be reformulated as 
\begin{align}
    \label{eq:KL_fix}
    &\mathrm{KL} \left( p (\mathbf{x}|\mathbf{y}) || g(\mathbf{x}|\mathbf{y}) \right) \nonumber \\
    &= \sum_{l=0}^{L} \mathrm{KL}(g \| f_l) - \sum_{l=1}^{L-1} \mathrm{KL}(g \| \tilde{f}_l) + \mathrm{H}(g) + \mathrm{const}.
\end{align}

Our final objective is to design the approximate posterior $g(\mathbf{x}|\mathbf{y})$ so as to minimize the KL divergence in \eqref{eq:KL_fix}.
However, directly solving this minimization problem over $g$ is intractable since the KL divergence involves multidimensional integrals with respect to $\mathbf{x}$.
To circumvent this difficulty, we introduce the \textit{auxiliary} PDFs $\{b_l \}_{l=0}^L, \{\tilde{b}_{l^\prime} \}_{l^\prime=1}^{L-1}, q$, which satisfy $b_l(\mathbf{x}) = b_{l^\prime}(\mathbf{x}) = q(\mathbf{x})$, 
$\forall l \in \{0,\ldots, L\}$, $\forall l^\prime \in \{1,\ldots, L-1\}$.
Using these definitions, the KL minimization of \eqref{eq:KL_fix} over $g$ can be reformulated as the equivalent minimization problem:
\begin{subequations}
\label{eq:opt_pdf}
\begin{align}
    \underset{ \{b_l \}_{l=0}^L, \{\tilde{b}_l\}_{l=1}^{L-1},q}{\mathrm{minimize}} \quad 
    & J \left ( \{b_l \}_{l=0}^L, \{\tilde{b}_l \}_{l=1}^{L-1}, q \right ) \\ 
    \mathrm{subject \ to} \quad 
    & b_0(\mathbf{x}) = \cdots = b_L(\mathbf{x}) = q(x) \\ 
    & \tilde{b}_1(\mathbf{x}) = \cdots = \tilde{b}_{L-1}(\mathbf{x}) = q(x),
\end{align}
\end{subequations}
where we define
\begin{align}
    &J \left ( \{b_l \}_{l=0}^L, \{\tilde{b}_l \}_{l=1}^{L-1}, q \right ) \nonumber \\
    &\triangleq \sum_{l=0}^L \mathrm{KL}(b_l \| f_l)  
    - \sum_{l=1}^{L-1} \mathrm{KL}(\tilde{b}_l \| \tilde{f}_l)
    + \mathrm{H}(q).
\end{align}

Since directly solving this optimization problem is intractable, we instead reformulate it as the following optimization problem with relaxed equality constraints as
\begin{subequations}
\label{eq:opt_moment}
\begin{align}
    & \underset{ \{b_l \}_{l=0}^L, \{\tilde{b}_l\}_{l=1}^{L-1},q}{\mathrm{minimize}} \quad 
    J \left ( \{b_l \}_{l=0}^L, \{\tilde{b}_l \}_{l=1}^{L-1}, q \right ) \\ 
    & \quad \mathrm{subject \ to} \quad \nonumber \\ 
    & \quad  \mathbb{E}_{b_l (\mathbf{x})} [\mathbf{x}] = \mathbb{E}_{q(\mathbf{x})} [\mathbf{x}],\ \forall l \in \{0, \ldots L \} \\
    & \quad \mathbb{E}_{\tilde{b}_l (\mathbf{x})} [\mathbf{x}] = \mathbb{E}_{q (\mathbf{x})} [\mathbf{x}],\ \forall l \in \{1, \ldots L-1 \} \\
    & \quad \mathbf{d} \left (  \mathbb{E}_{b_l (\mathbf{x})} [\mathbf{x} \mathbf{x}^\mathrm{H}] \right ) = \mathbf{d} \left ( \mathbb{E}_{q (\mathbf{x})} [\mathbf{x} \mathbf{x}^\mathrm{H} ] \right ), \forall l \in \{0, \ldots L \} \\ 
    & \quad \mathbf{d} \left (  \mathbb{E}_{\tilde{b}_l (\mathbf{x})} [\mathbf{x} \mathbf{x}^\mathrm{H}] \right ) = \mathbf{d} \left ( \mathbb{E}_{q (\mathbf{x})} [\mathbf{x} \mathbf{x}^\mathrm{H} ] \right ), \forall l \in \{1, \ldots L-1 \},
\end{align}
\end{subequations}
where these equality constraints require the first- and second-order moment-matching among the distributions $b_l, \tilde{b}_l, q$, rather than enforcing a perfect match of the distributions.

With respect to the relaxed minimization problem in \eqref{eq:opt_moment} and the fixed point in \eqref{eq:fixed_point}, the following proposition holds.

\begin{proposition}
    \label{proposition:fixed}
    Let $b_0, \{b_l\}_{l=1}^L, \{\tilde{b}_l\}_{l=1}^{L-1}, q$ denote the PDFs parameterized by 
    $\{\bar{\bm{\gamma}}^q, \bar{\bm{\mu}}^q \}$, 
    $\{ \bm{\gamma}^w_l, \bm{\mu}^w_l \}$, 
    $\{ \tilde{\bm{\gamma}}^w_l, \tilde{\bm{\mu}}^w_l \}$, and 
    $\{ \bm{\gamma}_\mathrm{f}, \mathbf{x}_\mathrm{f} \}$
    of the fixed point in \eqref{eq:fixed_point_2}, defined as 
    \begin{subequations}        
    \begin{align}
        \label{eq:b_0}
        b_0(\mathbf{x}) & \propto f_0(\mathbf{x}) \cdot \exp \left [ - \| \mathbf{x} - \bar{\boldsymbol{\mu}}^q \|_{\bar{\boldsymbol{\gamma}}^q}^2 \right ], \\
        \label{eq:b_l}
        b_l(\mathbf{x}) & \propto f_l (\mathbf{x}) \cdot \exp \left [ - \| \mathbf{x} - \boldsymbol{\mu}_l^w \|_{\boldsymbol{\gamma}_l^w}^2 \right ], \\
        \label{eq:b_l_tilde}
        \tilde{b}_l(\mathbf{x}) & \propto \tilde{f}_l (\mathbf{x}) \cdot \exp \left [- \| \mathbf{x} - \tilde{\boldsymbol{\mu}}_l^w \|_{\tilde{\boldsymbol{\gamma}}_l^w}^2 \right ], \\ 
        \label{eq:q}
        q(\mathbf{x}) & \propto \exp \left [ - \| \mathbf{x} - \mathbf{x}_\mathrm{f} \|_{\boldsymbol{\gamma}_\mathrm{f}}^2 \right ].
    \end{align}
    \end{subequations}
    Then, these PDFs correspond to the stationary point of the optimization problem under the moment-matching constraints defined in \eqref{eq:opt_moment}.
    Additionally, the mean and variance of these PDFs satisfy
    \begin{subequations}        
    \begin{align}
        & \mathbb{E}_{b_0(\mathbf{x})} [\mathbf{x}] 
        = \mathbb{E}_{b_l(\mathbf{x})} [\mathbf{x}] 
        = \mathbb{E}_{\tilde{b}_{l}(\mathbf{x})} [\mathbf{x}] 
        = \mathbb{E}_{q(\mathbf{x})} [\mathbf{x}] = \mathbf{x}_\mathrm{f}, \\
        & \mathbf{d} \left ( \mathbb{E}_{b_0(\mathbf{x})} [\mathbf{x} \mathbf{x}^\mathrm{H}] \right )
        = \mathbf{d} \left ( \mathbb{E}_{b_l(\mathbf{x})} [\mathbf{x} \mathbf{x}^\mathrm{H}] \right )
        = \mathbf{d} \left ( \mathbb{E}_{\tilde{b}_{l}(\mathbf{x})} [\mathbf{x} \mathbf{x}^\mathrm{H}] \right ) \nonumber \\
        &= \mathbf{d} \left ( \mathbb{E}_{q(\mathbf{x})} [\mathbf{x} \mathbf{x}^\mathrm{H}] \right ) 
        = \mathbf{1}_M \oslash \bm{\gamma}_\mathrm{f}.
    \end{align}
    \end{subequations}
\end{proposition}
\begin{proof}
    See Appendix~\ref{apx:fixed}.
\end{proof}

\subsection{Complexity Analysis}
\label{subsec:complexity}

This section evaluates computational complexity in terms of the number of complex multiplications, measured as \ac{FLOPs}.
The \ac{FLOPs} of the proposed and conventional algorithms are summarized in Table~\ref{table:FLOPs}.
In LMMSE-EP, a full-dimensional matrix multiplication and inversion is required at each algorithmic iteration.
Using the Woodbury formula~\cite{2013Gene_Matrix_Book}, the complexity order of these matrix operations per iteration is $\mathcal{O}( M^3 + M^2 N )$. 
Although this approach delivers superior performance by fully accounting for the correlation among all measurements, it incurs the highest computational complexity.
 
In contrast, both the proposed \textit{OvEP} and the conventional \textit{NOvEP} partition the large-scale measurements into smaller blocks, designing an LMMSE filter for each block.
Accordingly, the computational burden associated with matrix inversion can be alleviated owing to the reduced matrix size and the parallelizability of block-wise processing.
For \textit{OvEP}, the complexity order of the filter computation per iteration is $\mathcal{O}( L (M N_b^2 + N_b^3)
+ (L-1) (M \tilde{N}_\mathrm{b}^2 + \tilde{N}_\mathrm{b}^3 ) )$ since the measurements are partitioned into $L$ block parts of size $N_\mathrm{b}$ and $L-1$ overlapping parts of size $\tilde{N}_\mathrm{b}$.
For \textit{NOvEP}, where no overlap exists in the block partitioning, the complexity order of the filter computation per iteration is $\mathcal{O}(L(M N_\mathrm{b}^2 + N_\mathrm{b}^3 ))$. 
The complexity order of the denoiser function is identical across all methods, given by $\mathcal{O}(MQ)$.

A detailed comparison of the computational complexity under specific parameter settings will be presented in Section~\ref{sec:simulation}.

\begin{table}[t!]
    \caption{Computational complexity of the EP algorithms} \label{table:FLOPs}
    \centering
    \begin{tabular*}{8.5cm}{c|l}
        \hline
        \multicolumn{1}{c}{Algorithm} & FLOPs	\\
        \hline
        LMMSE-EP~\cite{2014Cespedes_EP_MIMO}
        & $\quad \mathcal{O}\left( T_\text{lmmse} \left ( M^3 + M^2N + MQ \right ) \right)$ \\
        MF-EP~\cite{2015Meng_MFEP}
        & $\quad \mathcal{O} \left ( T_\text{mf} \left ( MN + MQ \right ) \right )$ \\
        NOvEP~\cite{2020Wang_EP_subarray}
        & $\quad \mathcal{O} \left ( T_\text{nov} \left ( L (M N_b^2 + N_b^3) + MQ \right ) \right )$ \\
        OvEP (Prop.) & 
        \begin{tabular}{l}
            $\ \mathcal{O} \Big ( T_\text{ov} \Big ( L (M N_b^2 + N_b^3) $ \\
            $\quad \quad + (L-1) (M \tilde{N}_\mathrm{b}^2 + \tilde{N}_\mathrm{b}^3 ) + MQ \Big ) \Big )$ 
        \end{tabular} \\
        \hline
    \end{tabular*}
    \vspace{1mm}
    \\{Note: $T_\mathrm{lmmse}, T_\mathrm{mf}, T_\mathrm{nov}, T_\mathrm{ov}$ are the number of iterations}.
    \vspace{-4ex}
\end{table}

\section{Numerical Results}
\label{sec:simulation}

\begin{figure*}[t]
    \centering
    %
    \begin{minipage}{\figsize \columnwidth}
        \centering
        \includegraphics[width=\linewidth]{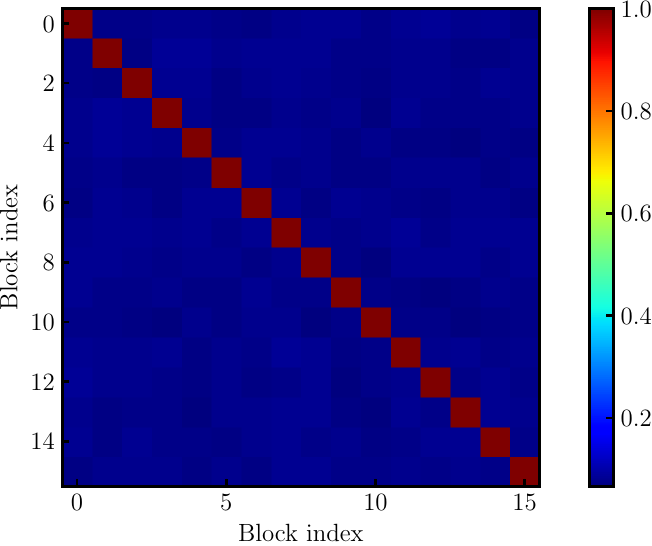}
        \subcaption{NOvEP $(N_\mathrm{s}=2, L=16)$}
        \label{fig:Cmap_NOvEP_r0}
    \end{minipage}
    %
    \begin{minipage}{\figsize \columnwidth}
        \centering
        \includegraphics[width=\linewidth]{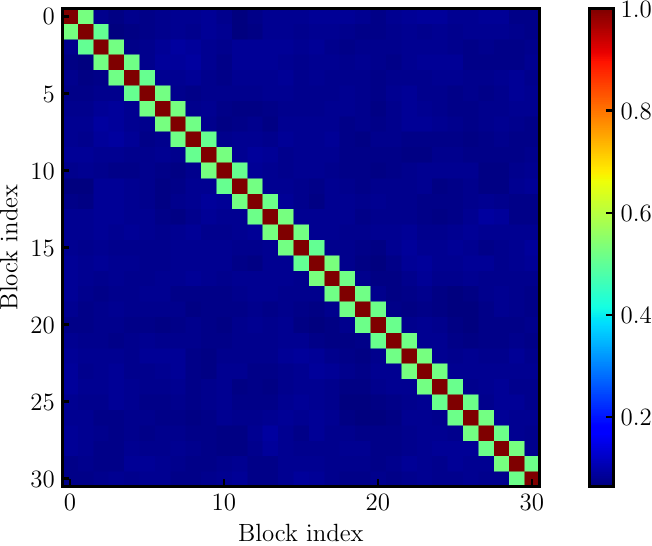}
        \subcaption{OvEP w/o subtraction $(N_\mathrm{s}=1, L=31)$}
        \label{fig:Cmap_OvEP_NonSubtract_r0}
    \end{minipage}
    %
    \begin{minipage}{\figsize \columnwidth}
        \centering
        \includegraphics[width=\linewidth]{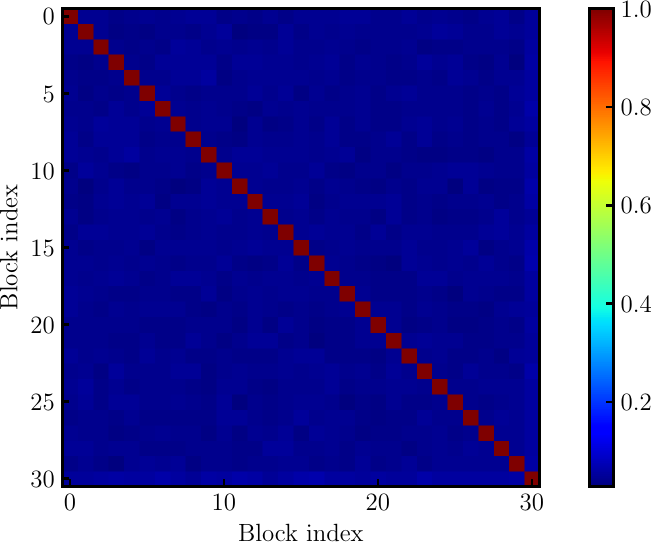}
        \subcaption{OvEP $(N_\mathrm{s}=1, L=31)$}
        \label{fig:Cmap_OvEP_r0}
    \end{minipage}
    \caption{Absolute value of the inter-block correlation matrix $\mathbf{\Phi}$ in uncorrelated channels ($\rho=0$) with $N=32, M=24, Q=4,$ and $N_\mathrm{b}=2$.}
    \label{fig:Cmap_r0}
    \vspace{0.3cm}
    %
    \begin{minipage}{\figsize \columnwidth}
        \centering
        \includegraphics[width=\linewidth]{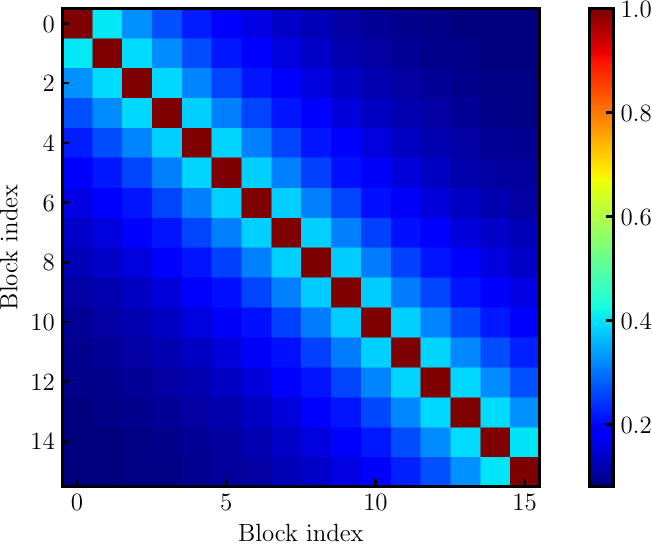}
        \subcaption{NOvEP $(N_\mathrm{s}=2, L=16)$}
        \label{fig:Cmap_NOvEP_r095}
    \end{minipage}
    %
    \begin{minipage}{\figsize \columnwidth}
        \centering
        \includegraphics[width=\linewidth]{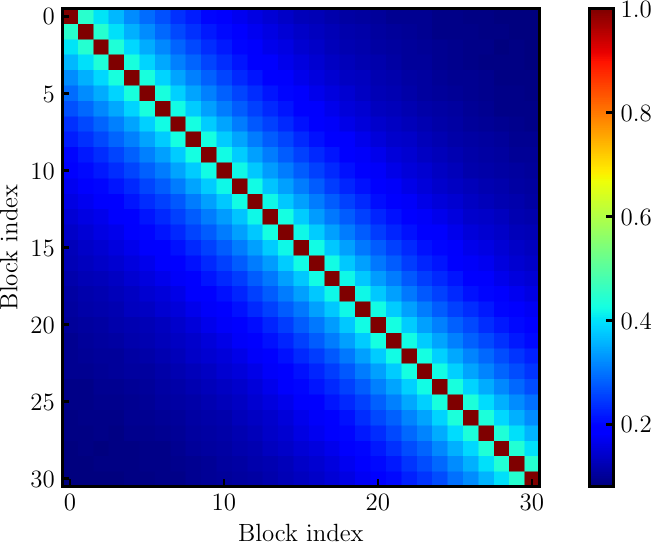}
        \subcaption{OvEP w/o subtraction $(N_\mathrm{s}=1, L=31)$}
        \label{fig:Cmap_OvEP_NonSubtract_r095}
    \end{minipage}
    %
    \begin{minipage}{\figsize \columnwidth}
        \centering
        \includegraphics[width=\linewidth]{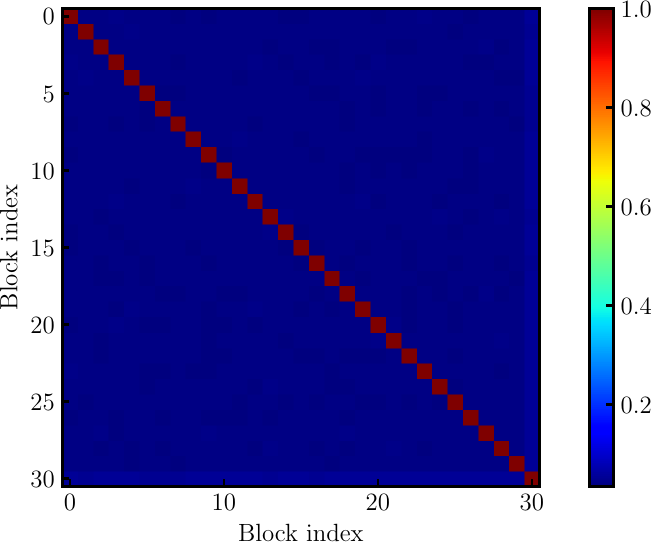}
        \subcaption{OvEP $(N_\mathrm{s}=1, L=31)$}
        \label{fig:Cmap_OvEP_r095}
    \end{minipage}
    %
    \caption{Absolute value of the inter-block correlation matrix $\mathbf{\Phi}$ in highly correlated channels ($\rho=0.95$) with $N=32, M=24, Q=4,$ and $N_\mathrm{b}=2$.}
    \label{fig:Cmap_r095}
    \vspace{-3ex}
\end{figure*}

\begin{figure}[t]
    \centering
    \includegraphics[width=\linewidth]{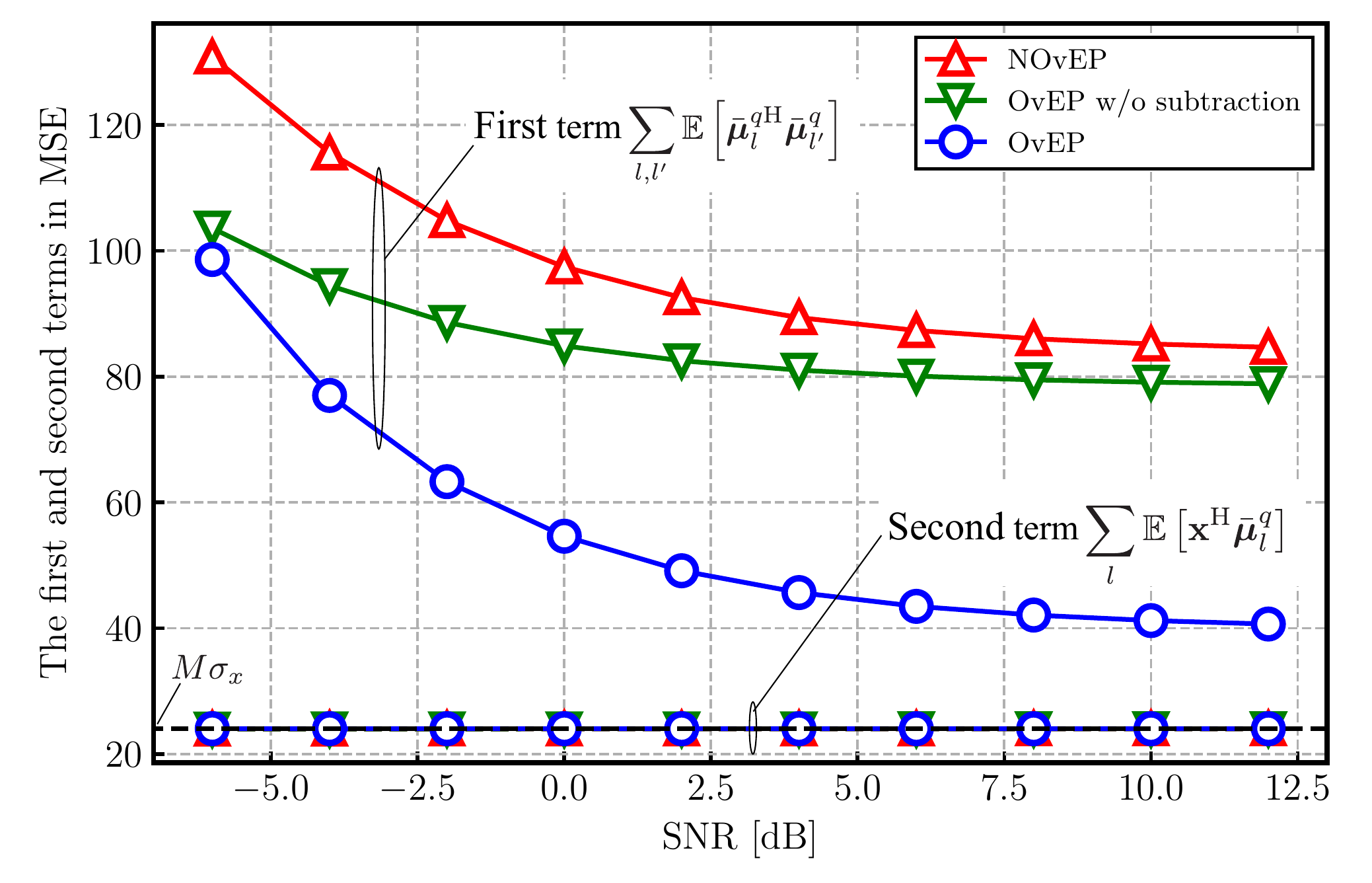}
    \vspace{-4ex}
    \caption{
    The first term $\sum_{l,l^\prime} \mathbb{E} [ \bar{\bm{\mu}}^{q\mathrm{H}}_l \bar{\bm{\mu}}^{q}_{l^\prime}]$ and the second term $\sum_{l} \mathbb{E} \left [ \mathbf{x}^\mathrm{H} \bar{\bm{\mu}}_l^q \right ]$ of the MSE in \eqref{eq:MSE_in} at the first iteration ($t=1)$ under various SNR values with $N=32, M=24, \rho=0.95,$ and $Q=4$.} 
    \label{fig:MSE_vs_SNR}
    \vspace{-3ex}
\end{figure}

\begin{figure*}[t]
    \centering
    \begin{minipage}{\figsize \columnwidth}
        \centering
        \includegraphics[width=\linewidth]{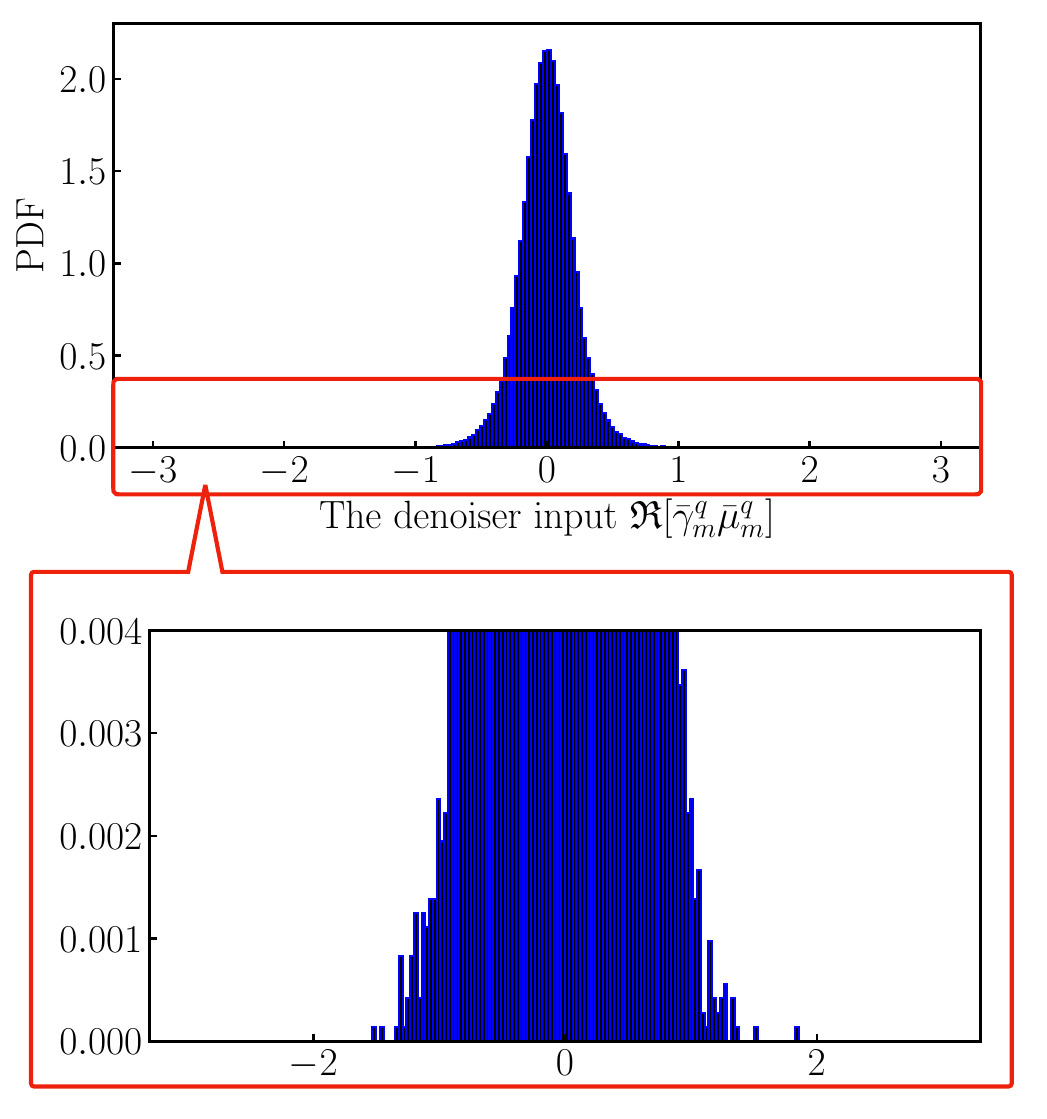}
        \subcaption{NOvEP $(N_\mathrm{s}=2, L=16)$}
        \label{fig:Mean_nrm_In_NOvEP}
    \end{minipage}
    %
    \begin{minipage}{\figsize \columnwidth}
        \centering
        \includegraphics[width=\linewidth]{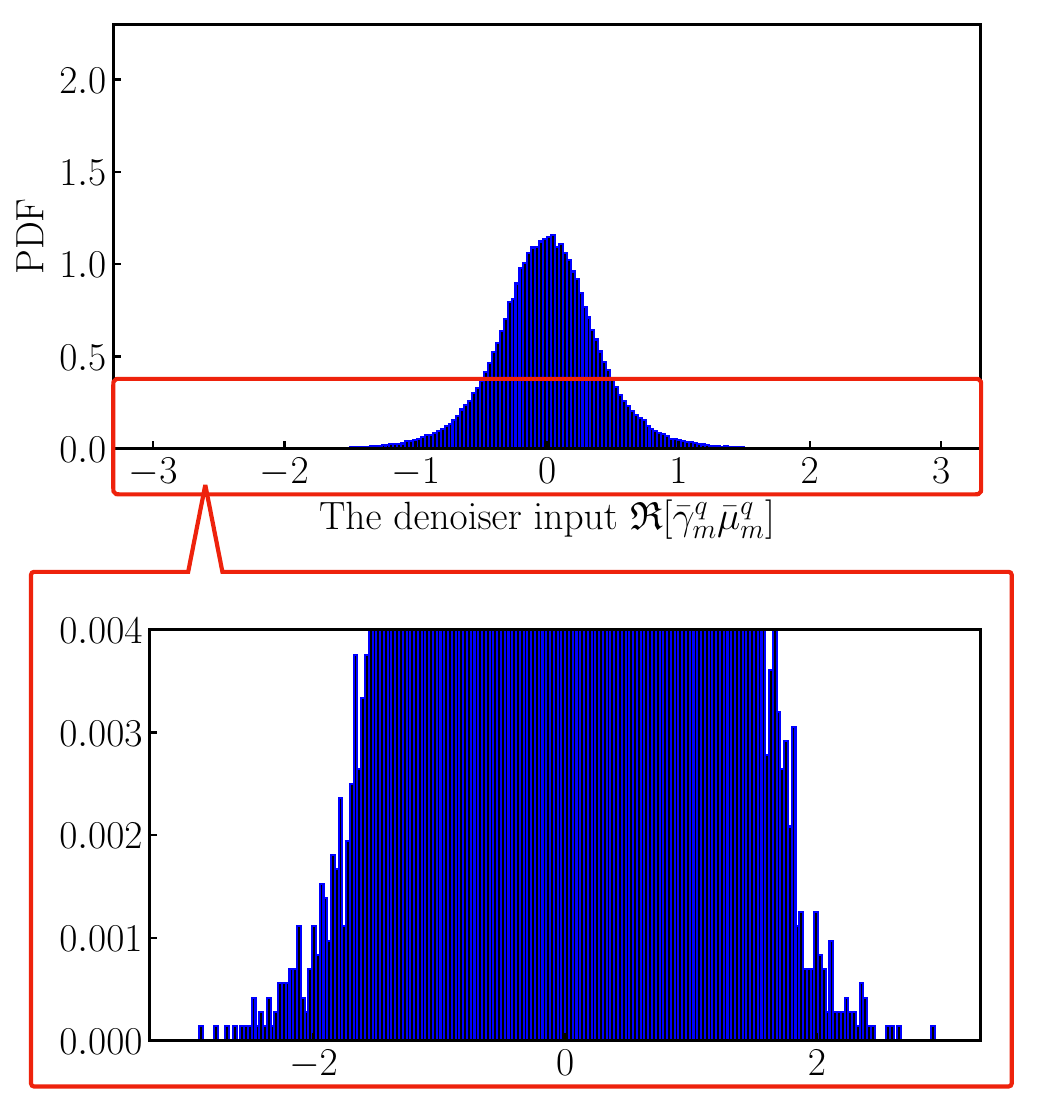}
        \subcaption{OvEP w/o subtraction $(N_\mathrm{s}=1, L=31)$}
        \label{fig:Mean_nrm_In_OvEP_NonSubtract}
    \end{minipage}
    %
    \begin{minipage}{\figsize \columnwidth}
        \centering
        \includegraphics[width=\linewidth]{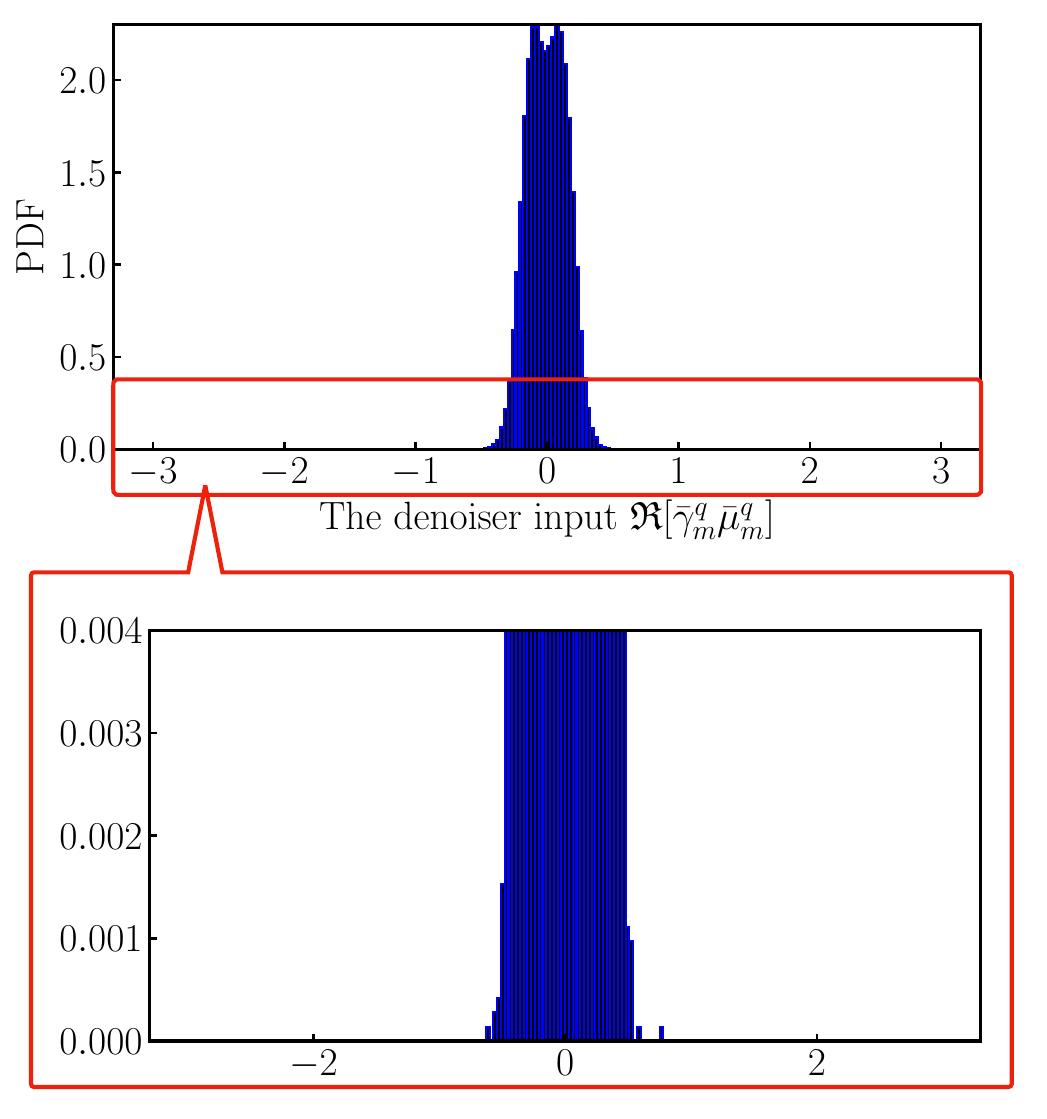}
        \subcaption{OvEP $(N_\mathrm{s}=1, L=31)$}
        \label{fig:Mean_nrm_In_OvEP}
    \end{minipage}
    \caption{Histogram of the denoiser input $\mathfrak{R} [ \bar{\gamma}_m^q \bar{\mu}_m^q ]$ at the initial iteration $(t=1)$ with $Q=4, N_\mathrm{b}=2, \rho=0.95,$ and $\mathrm{SNR}=12\ \mathrm{dB}$.}
    \label{fig:Hist_in}
    \vspace{-1ex}
\end{figure*}

\begin{figure*}[t]
    \centering
    \begin{minipage}{\figsize \columnwidth}
        \centering
        \includegraphics[width=\linewidth]{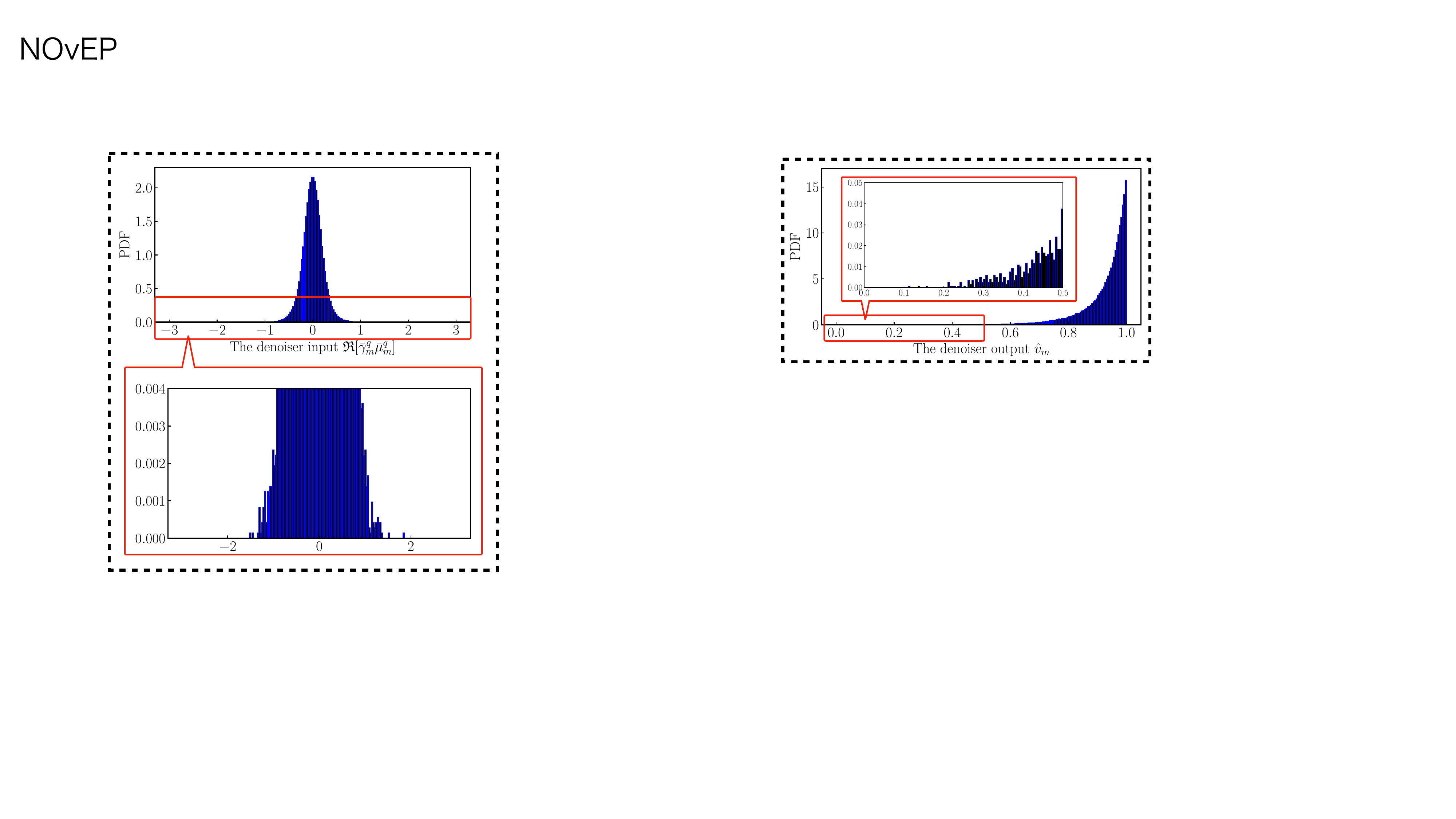}
        \subcaption{NOvEP $(N_\mathrm{s}=2, L=16)$}
        \label{fig:Vari_nrm_In_NOvEP}
    \end{minipage}
    %
    \begin{minipage}{\figsize \columnwidth}
        \centering
        \includegraphics[width=\linewidth]{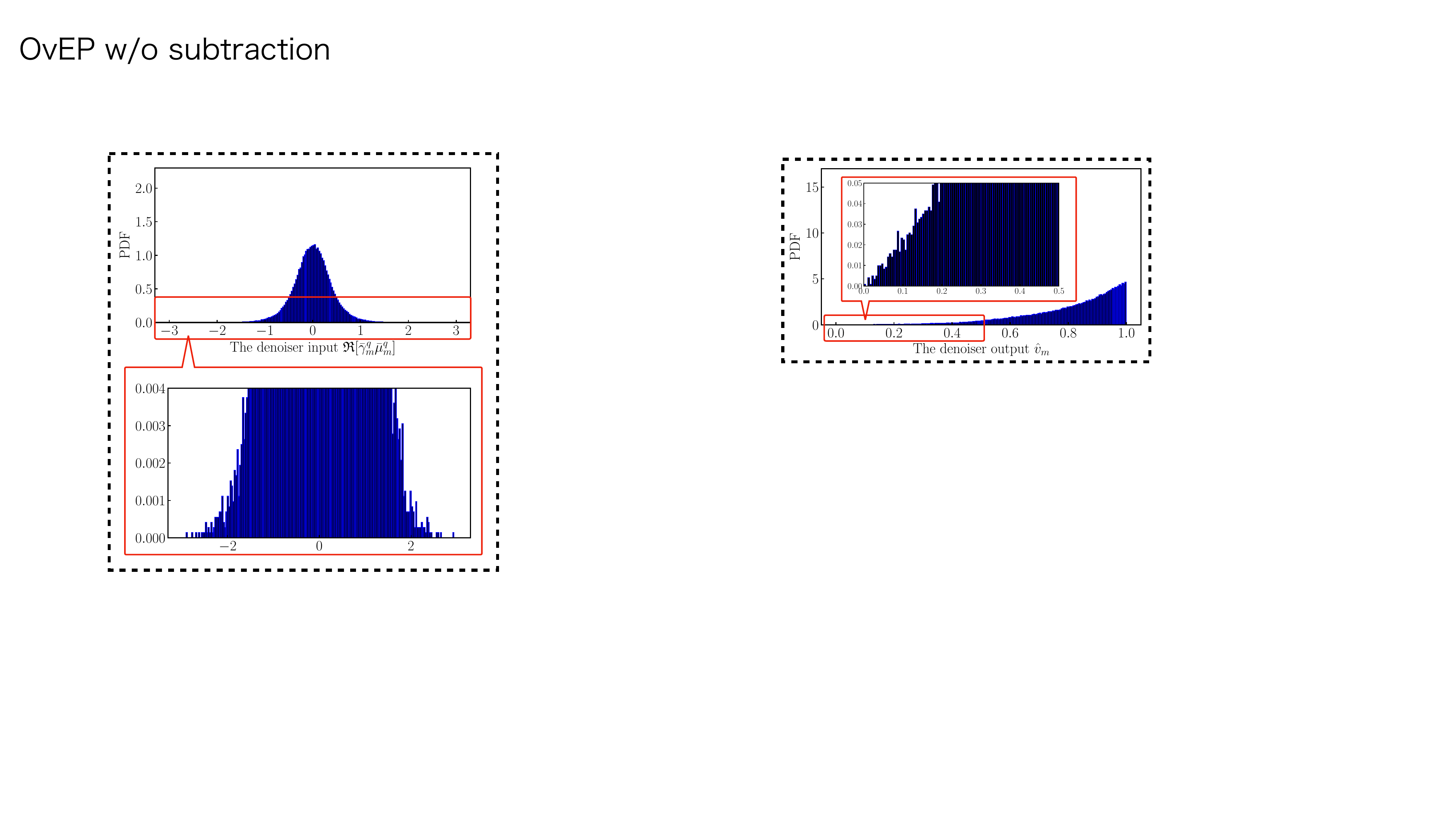}
        \subcaption{OvEP w/o subtraction $(N_\mathrm{s}=1, L=31)$}
        \label{fig:Vari_nrm_In_OvEP_NonSubtract}
    \end{minipage}
    %
    \begin{minipage}{\figsize \columnwidth}
        \centering
        \includegraphics[width=\linewidth]{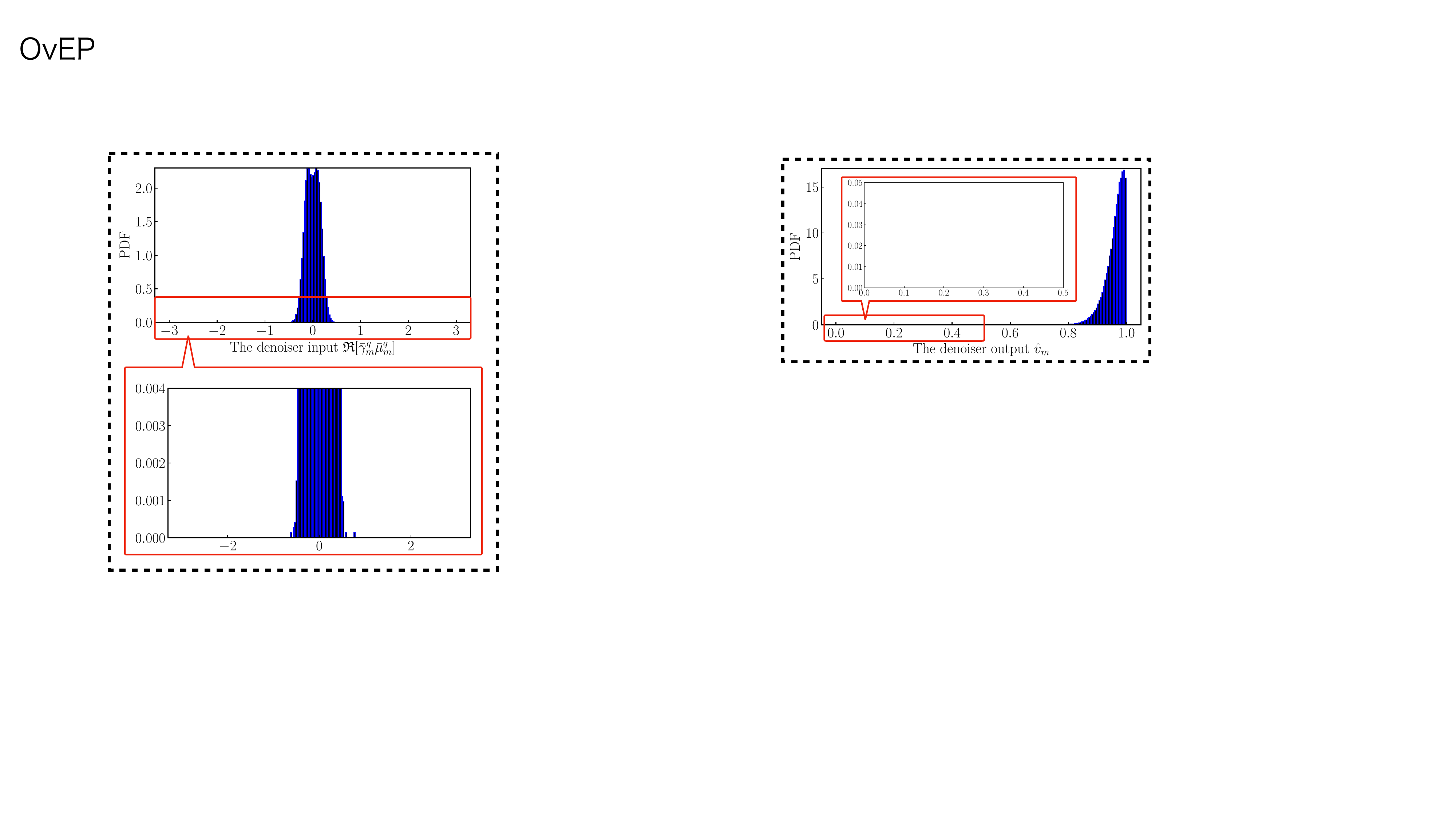}
        \subcaption{OvEP $(N_\mathrm{s}=1, L=31)$}
        \label{fig:Vari_nrm_In_OvEP}
    \end{minipage}
    \caption{Histogram of the variance of the denoiser output $\hat{v}_m$ at the initial iteration $(t=1)$ with $Q=4, N_\mathrm{b}=2, \rho=0.95,$ and $\mathrm{SNR}=12\ \mathrm{dB}$.}
    \label{fig:Hist_out}
    \vspace{-3ex}
\end{figure*}

\subsection{Simulation Setup}
\label{subsec:sim_set}

This section evaluates the performance of the proposed and conventional algorithms under the following setup.
The spatially correlated channel matrix $\mathbf{H} \in \mathbb{C}^{N \times M}$ is generated based on the Kronecker model, as in \cite{2018He_OAMPNet,2020He_OAMPNet2, 2024Furudoi_ADD}:
\begin{align*}
    \mathbf{H} = \mathbf{R}_\mathrm{r}^{1/2} \mathbf{G} \mathbf{R}_\mathrm{t}^{1/2},
\end{align*}
where $\mathbf{R}_\mathrm{r} \in \mathbb{C}^{N \times N}$ and $\mathbf{R}_\mathrm{t} \in \mathbb{C}^{M \times M}$ denote the spatial correlation matrices at the BS and the UE sides, respectively.
Each element of $\mathbf{G} \in \mathbb{C}^{N \times M}$ follows i.i.d. complex Gaussian distribution $\mathcal{CN}(0, 1)$.
The spatial correlation matrix at the UE side is set to $\mathbf{R}_\mathrm{t} = \mathbf{I}_M$, assuming that each UE is independently distributed within the coverage area.
In contrast, the spatial correlation matrix at the BS side, $\mathbf{R}_\mathrm{r}$, is generated based on the exponential attenuation model~\cite{2001Loyka_KronExp}, given by
$[\mathbf{R}_\mathrm{r}]_{n,n^\prime} = \rho^{|n - n^\prime|},$
where $\rho \in [0, 1]$ denotes the spatial correlation coefficient between receiver antennas at the BS.

In this simulation, the following methods are compared:
\begin{itemize}
    \item[1)] \textbf{\textit{LMMSE}}: 
    The classical signal detection method.
    \item[2)] \textbf{\textit{LMMSE-EP}}~\cite{2014Cespedes_EP_MIMO}: The traditional EP-based detector, as described in Section~\ref{subsec:Algorithm_Description}.
    \item[3)] \textbf{\textit{NOvEP}}~\cite{2020Wang_EP_subarray}: The low-complexity variant of the EP-based detector employing non-overlapping block partitioning, as described in Section~\ref{subsec:Algorithm_Description}.
    \item[4)] \textbf{\textit{OvEP w/o subtraction}}: The proposed algorithm without subtracting the overlapping parts from the block parts, i.e., $\tilde{\bm{\mu}}_{l}^q = \tilde{\bm{\gamma}}_{l}^q = \mathbf{0}_M$ in~\eqref{eq:mu_bar_q}. 
    This method is evaluated to assess the impact of the subtraction operation in \eqref{eq:mu_bar_q}, which is an essential step of the proposed approach.
    \item[5)] \textbf{\textit{OvEP}}: The proposed algorithm presented in Algorithm~\ref{alg:OvEP}.
\end{itemize}

Unless otherwise specified, the following parameters are used in the simulations; $N=32,\ M=24,\ Q \in \{4,16\},\ \rho \in \{0,0.95\},\ \sigma_x=1,\ \beta_x=10$.
For the proposed methods, \textit{OvEP} and \textit{OvEP w/o subtraction}, the block size and shift size are set to $N_\mathrm{b}=2$ and $N_\mathrm{s}=1$, respectively, yielding an overlapping part of size $\tilde{N_\mathrm{b}}=1$ and $L=31$ blocks.
For \textit{NOvEP}, these parameters are set to $N_\mathrm{b}=2$ and $N_\mathrm{s}=2$, resulting in no overlap, i.e., $\tilde{N_\mathrm{b}}=0$ and $L=16$.
The parameter setting $N_\mathrm{b}=2$ corresponds to the smallest block size in overlapping partitioning, with the lowest computational complexity.
The number of iterations is set to $T_\mathrm{lmmse}=16, T_\mathrm{nov}=T_\mathrm{ov}=32$ for \textit{LMMSE-EP}, \textit{NOvEP}, and \textit{OvEP}, respectively.
To enhance convergence performance, a damping scheme~\cite{2019Rangan_AMP_damping,2001Minka_EP_thesis,2014Cespedes_EP_MIMO} is introduced in lines 18-21 of Algorithm~\ref{alg:OvEP}, with the damping factor fixed at $0.5$ for all methods.

\subsection{Analysis of the Proposed Algorithm}

As stated in Proposition~\ref{proposition:mse}, the MSE of the denoiser input $\bar{\bm{\mu}}^q$ is determined by the inter-block correlation.
This section demonstrates that subtracting the overlapping part, as expressed in \eqref{eq:mu_bar_q}, effectively reduces the inter-block correlation, thereby improving the MSE performance.

To quantify the inter-block correlation across all block combinations, we introduce the normalized inter-block correlation matrix $\mathbf{\Phi} \in \mathbb{C}^{L \times L}$, whose $(l,l^\prime)$-th element is defined as
\begin{align}
    \left [\mathbf{\Phi} \right ]_{l,l^\prime} \triangleq \frac{ \mathbb{E} \left [ \bar{\boldsymbol{\mu}}_l^{q\mathrm{H}} \bar{\boldsymbol{\mu}}_{l^\prime}^q \right ] }{ \mathbb{E} \left [ \| \bar{\boldsymbol{\mu}}_l^{q} \|_2 \right ] \mathbb{E} \left [ \| \bar{\boldsymbol{\mu}}_{l^\prime}^{q} \|_2 \right ]}.
\end{align}

Figs.~\ref{fig:Cmap_r0} and \ref{fig:Cmap_r095} show the inter-block correlation matrices for \textit{NOvEP}, \textit{OvEP w/o subtraction}, and \textit{OvEP} at the initial iteration $(t=1)$.
In the case of uncorrelated channels ($\rho=0$), all approaches exhibit low inter-block correlation.
By contrast, when the spatial correlation is high ($\rho=0.95$), the inter-block correlations of \textit{NOvEP} and \textit{OvEP w/o subtraction} increase significantly, whereas the proposed \textit{OvEP} effectively suppresses the inter-block correlation.
The comparison between \textit{OvEP} and \textit{OvEP w/o subtraction} demonstrates that the subtraction operation in \eqref{eq:mu_bar_q} plays a key role in reducing inter-block correlation.

Fig.~\ref{fig:MSE_vs_SNR} shows the first term $\sum_{l,l^\prime} \mathbb{E} [ \bar{\bm{\mu}}^{q\mathrm{H}}_l \bar{\bm{\mu}}^{q}_{l^\prime}]$ and the second term $\sum_{l} \mathbb{E} \left [ \mathbf{x}^\mathrm{H} \bar{\bm{\mu}}_l^q \right ] $ of the MSE in~\eqref{eq:MSE_in} at the initial iteration under various $\mathrm{SNR} \triangleq \sigma_x/ \sigma_z$.
The first term corresponds to the inter-block correlation, whereas the second term corresponds to the alignment of $\bar{\bm{\mu}}^q$ with the target vector $\mathbf{x}$.
As illustrated in Fig.~\ref{fig:MSE_vs_SNR}, the second term remains close to the constant value $\sum_{l} \mathbb{E} \left [ \mathbf{x}^\mathrm{H} \bar{\bm{\mu}}_l^q \right ] \simeq \sigma_x M$, which demonstrates the validity of Proposition~\ref{proposition:mse} and the approximation in~\eqref{eq:align_x_2}.
Furthermore, the proposed method, through the subtraction of overlapping parts, effectively reduces the first term associated with inter-block correlation, thereby improving the MSE performance.

To assess the impact of MSE reduction at the initial iteration, we analyze the statistical behavior of the denoiser input and output.
In the case of \ac{QPSK}, i.e., $Q\!=\!4$, the mean $\hat{x}_m$ and variance $\hat{v}_m$ of the denoiser output are obtained from~\eqref{eq:mu_sigma_denoiser} as 
\begin{subequations}
\label{eq:denoiser_qpsk}
\begin{align}
    \hat{x}_m &= \eta_\mathrm{e} (\bar{\mu}^q_m, \bar{\gamma}_m^q) 
    =&& \sqrt{\frac{\sigma_x}{2}} \Big \{ 
    \mathrm{tanh} \left ( \sqrt{2 \sigma_x}\mathfrak{R} \left [ \bar{\gamma}_m^q  \bar{\mu}^q_m \right ] \right ) \nonumber \\
    &&& + j \ \mathrm{tanh} \left ( \sqrt{2 \sigma_x} \mathfrak{I} \left [ \bar{\gamma}_m^q \bar{\mu}^q_m \right ] \right ) \Big \} \\
    \hat{v}_m &= \eta_\mathrm{v} (\bar{\mu}^q_m, \bar{\gamma}_m^q) =&& \sigma_x - |\hat{x}_m|^2,
\end{align}
\end{subequations}

As can be seen from \eqref{eq:denoiser_qpsk}, the denoiser outputs $\{\hat{x}_m, \hat{v}_m \}$ are determined by the input $\bar{\gamma}_m^q \bar{\mu}_m^q$. 
When this input takes a large value, the denoiser operates in the inactive region~\cite{2024Furudoi_ADD}, where the denoiser output $\hat{x}_m$ is generated by a hard-decision that selects one symbol from the constellation points $\mathcal{X}=\{\mathrm{x}_1, \ldots, \mathrm{x}_Q\}$.
When such large inputs arise erroneously without reflecting the true reliability, they are regarded as \textit{outliers}.
In this case, the output variance $\hat{v}_m$ becomes very small, approaching zero.
If these small variances are generated at the initial iteration stage, erroneously hard-decided symbols $\hat{x}_m$ with spuriously high confidence propagate through subsequent iterations, leading to error propagation.

To validate the effectiveness of the proposed subtraction mechanism,  
Figs.~\ref{fig:Hist_in} and \ref{fig:Hist_out} show the histogram of the denoiser input $\mathfrak{R}[\bar{\gamma}_m^q \bar{\mu}_m^q]$ and the variance of the denoiser output $\hat{v}_m$ at the initial iteration $(t=1)$.
As shown in Fig.~\ref{fig:Hist_in}, in \textit{NOvEP} and \textit{OvEP w/o subtraction}, large values of the denoiser input arise due to the combination of highly correlated LMMSE outputs in lines 13-14 of Algorithm~\ref{alg:OvEP}. 
As a result, the denoiser operates in the inactive region, producing small variances at the initial iteration, as illustrated in Fig.~\ref{fig:Hist_out}.
In contrast, Fig.~\ref{fig:Mean_nrm_In_OvEP} demonstrates that the proposed method suppresses the occurrence of outliers in the denoiser input by reducing inter-block correlation through the subtraction of overlapping parts in~\eqref{eq:mu_bar_q}.
Consequently, excessively small variances in the denoiser output are avoided at the initial iteration, thereby preventing error propagation in subsequent iterations.

\vspace{-2mm}
\subsection{BER Performance}
\label{subsec:BER}

Figs.~\ref{fig:BER_vs_SNR} and \ref{fig:BER_vs_Iterations} show the \ac{BER} performance with respect to $\mathrm{SNR} = \sigma_x/ \sigma_z$ and the number of iterations, respectively.
As shown in Fig.~\ref{fig:BER_vs_SNR_r0_QPSK}, in the case of no spatial correlation $(\rho=0)$, all low-complexity EP-based methods achieve performance close to that of \textit{LMMSE-EP}.
However, as shown in Fig.~\ref{fig:BER_vs_SNR_r095_QPSK}, under high spatial correlation $(\rho=0.95)$, \textit{NOvEP} and \textit{OvEP w/o subtraction} exhibit severe BER degradation due to their inability to capture the spatial correlation.
In contrast, the proposed method can improve BER performance and approaches that of \textit{LMMSE-EP} by effectively reducing the inter-block correlation, as illustrated in Figs.~\ref{fig:Cmap_r0} and \ref{fig:Cmap_r095}.
Furthermore, as shown in Fig.~\ref{fig:BER_vs_SNR_r095_16QAM}, under higher-order modulation $(Q=16)$, the proposed method achieves a substantial BER reduction, whereas the conventional methods fail to detect the signals reliably.

\vspace{-2mm}
\subsection{Computational Complexity}
\label{subsec:flops}

Finally, to evaluate the computational complexity of the proposed and conventional algorithms, Fig.~\ref{fig:FLOPs_vs_N} presents the FLOPs, defined as the number of complex multiplications, versus the number of antennas $N$ with the compression rate fixed at $\alpha \triangleq N/M = 1.33$.
As summarized in Table~\ref{table:FLOPs}, \textit{LMMSE-EP} requires full-dimensional matrix operation with a complexity order of $\mathcal{O}(M^3 + M^2N)$ per iteration, whereas the proposed method requires $\mathcal{O}( L (M N_b^2 + N_b^3) + (L-1) (M \tilde{N}_\mathrm{b}^2 + \tilde{N}_\mathrm{b}^3 ) )$ due to the partitioning of the large-scale measurement into small-size blocks.
As depicted in Fig.~\ref{fig:FLOPs_vs_N}, the computational complexity of all methods increases with the number of antennas $N$.
However, the block-partition-based methods, \textit{OvEP} and \textit{NOvEP}, alleviate this increase compared to \textit{LMMSE-EP}.
These results confirm that \textit{OvEP} can significantly improve BER performance with only a slight increase in computational complexity compared to \textit{NOvEP}.

\begin{figure}[t]
    \begin{minipage}{1.0\columnwidth}
        \centering
        \includegraphics[width=\linewidth]{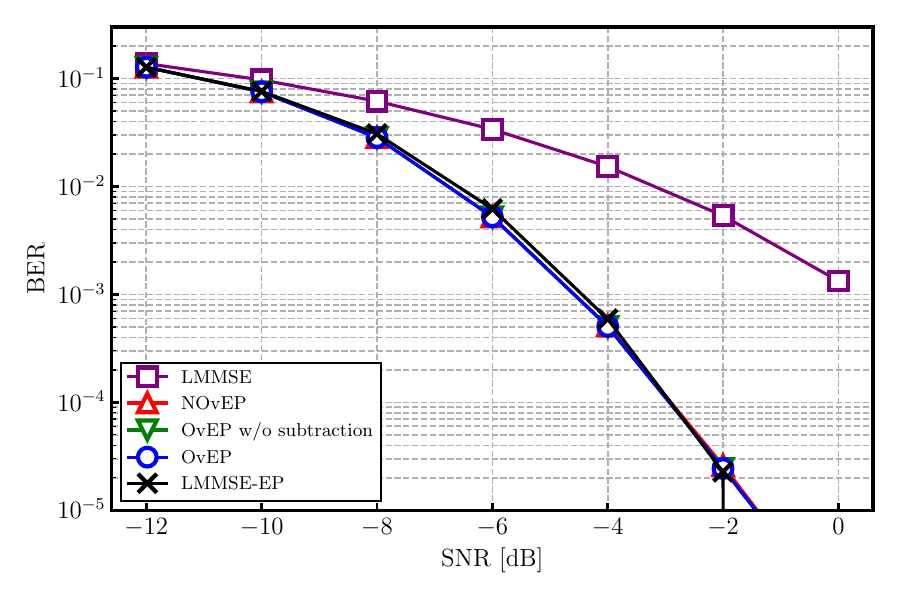}
        \vspace{-5ex}
        \subcaption{Uncorrelated channels ($\rho=0$) with $Q=4$.} 
        \label{fig:BER_vs_SNR_r0_QPSK}
    \end{minipage}
    %
    \begin{minipage}{1.0\columnwidth}
        \centering
        \includegraphics[width=\linewidth]{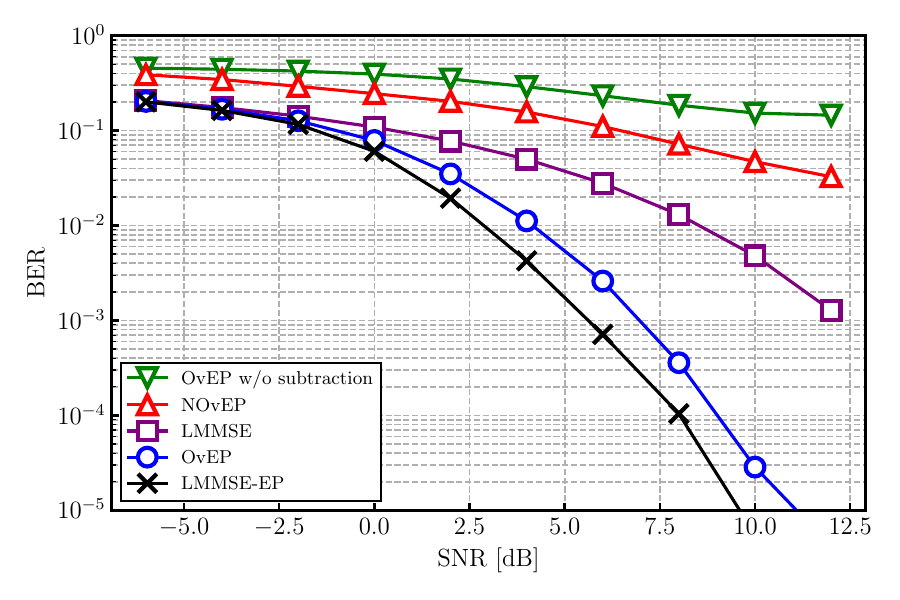}
        \vspace{-5ex}
        \subcaption{Highly correlated channels ($\rho=0.95$) with $Q=4$.} 
        \label{fig:BER_vs_SNR_r095_QPSK}
    \end{minipage}
    %
    \begin{minipage}{1.0\columnwidth}
        \centering
        \includegraphics[width=\linewidth]{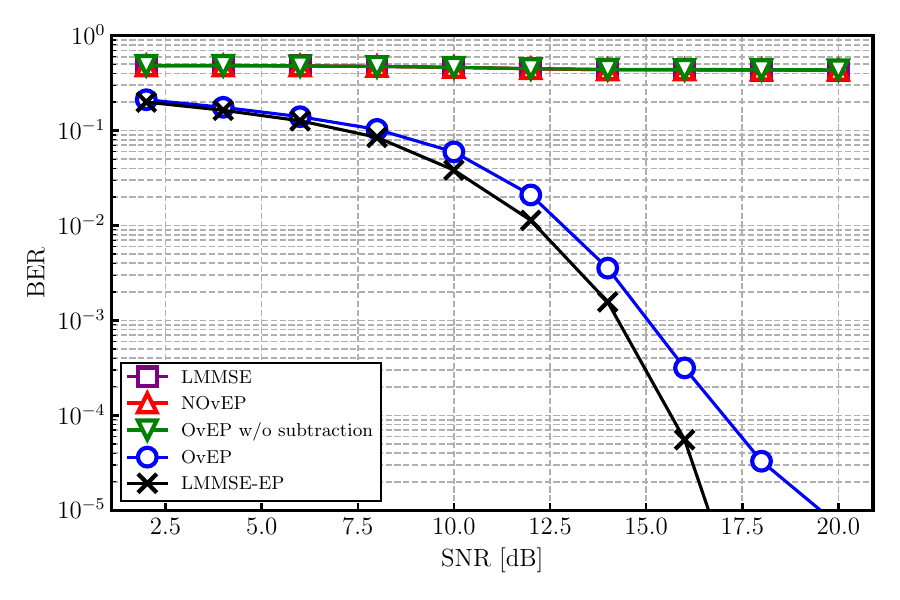}
        \vspace{-5ex}
        \subcaption{Highly correlated channels ($\rho=0.95$) with $Q=16$.} 
        \label{fig:BER_vs_SNR_r095_16QAM}
    \end{minipage}
    \caption{BER versus SNR with $N=32$ and $M=24$.} 
    \label{fig:BER_vs_SNR}
    \vspace{-3ex}
\end{figure}

\begin{figure}[t]
    \begin{minipage}{1.0\columnwidth}
        \centering
        \includegraphics[width=\linewidth]{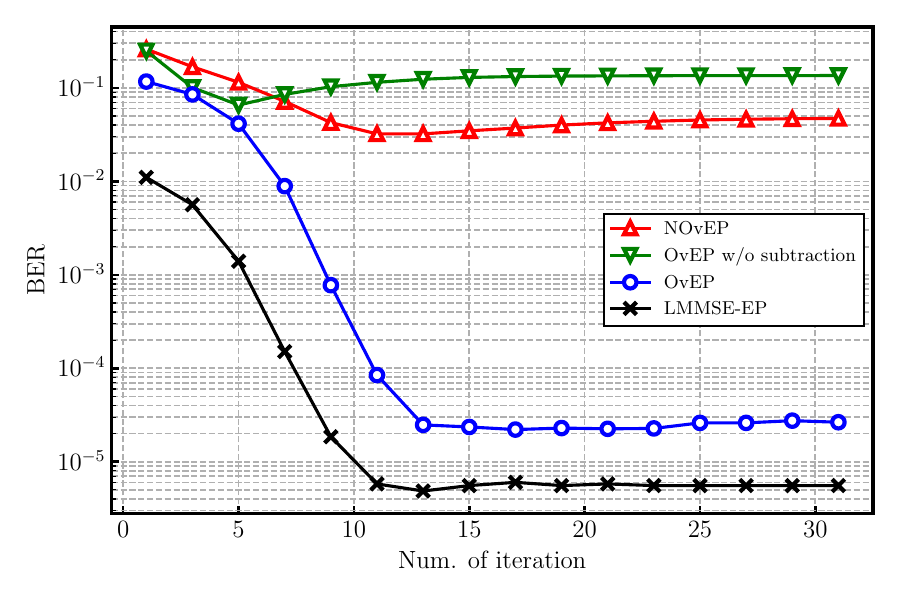}
        \vspace{-5ex}
    \end{minipage}
    %
    %
    \caption{BER versus iterations with $N=32, M=24, \rho=0.95,$ and $Q=4$.} 
    \label{fig:BER_vs_Iterations}
    \vspace{-3ex}
\end{figure}

\begin{figure}[t]
    \centering
    \includegraphics[width=\linewidth]{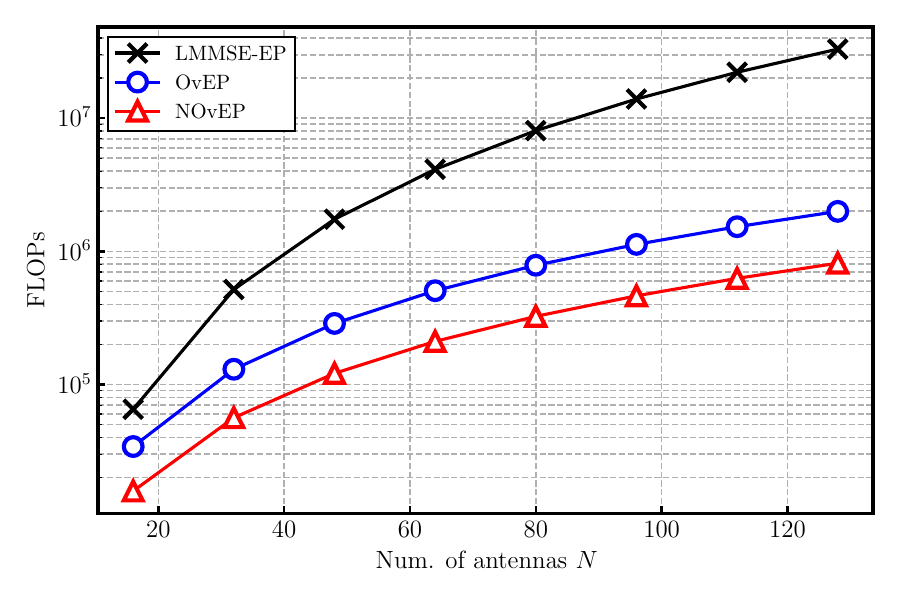}
    \vspace{-5ex}
    \caption{FLOPs versus the number of antennas $N$ with the compression rate fixed at $\alpha \triangleq N/M = 1.33$ and $N_\mathrm{b}=2$.}
    \label{fig:FLOPs_vs_N}
    \vspace{-4ex}
\end{figure}

\section{Conclusion}

This paper proposed an EP-based detector, termed \textit{OvEP}, for highly correlated MIMO systems.
In the proposed method, the large-scale measurements are partitioned into overlapping blocks, where low-complexity LMMSE-based filters are designed for both block and overlapping parts.
By subtracting the LMMSE outputs of the overlapping parts from those of the block parts, the inter-block correlation induced by highly correlated channels can be effectively reduced.
Theoretical analysis established that the MSE of the denoiser input is determined by the inter-block correlation among block-wise LMMSE outputs.
Furthermore, it was theoretically proven that the fixed point of \textit{OvEP} corresponds to a stationary point of a relaxed version of the KL minimization problem. 
Comprehensive numerical results validated that \textit{OvEP} achieves substantial performance improvements over the conventional \textit{NOvEP} in highly correlated channels, with only a slight and manageable increase in computational complexity. 


\appendices
\section{Proof of Proposition~\ref{proposition:mse}}
\label{apx:mse}

From \eqref{eq:mu_bar_q}, $\bar{\bm{\mu}}_l^q$ can be expressed as
\begin{align}
    \label{eq:mu_bar_q_l}
    \bar{\bm{\mu}}_l^q
    =\mathbf{D}(\bar{\boldsymbol{\gamma}}^q)^{-1} \left \{
    \mathbf{D} (\boldsymbol{\gamma}_l^q) \boldsymbol{\mu}_l^q - 
    \mathbf{D} (\tilde{\boldsymbol{\gamma}}_l^q) \tilde{\boldsymbol{\mu}}_l^q \right \}.
\end{align}

Substituting \eqref{eq:mu_bar_q_l} into \eqref{eq:MSE_in}, 
the second term of the MSE in \eqref{eq:MSE_in} can be expressed as
\begin{align}
    &\sum_{l=1}^L \mathbb{E}_{p(\mathbf{H}, \mathbf{x}, \mathbf{z})} \left [ \mathbf{x}^\mathrm{H} \bar{\boldsymbol{\mu}}_l^q \right ]
    = \sum_{l=1}^L \mathrm{tr} \left ( \mathbb{E}_{p(\mathbf{H}, \mathbf{x}, \mathbf{z})} \left [  \bar{\boldsymbol{\mu}}_l^q \mathbf{x}^\mathrm{H} \right ] \right ) \nonumber
\end{align}
\begin{align}
    \label{eq:align_x}
    &= \mathbb{E}_{p(\mathbf{H})} \Bigg [
    \sum_{l=1} ^L
    \mathrm{tr} \Bigg \{
    \mathbf{D}(\bar{\boldsymbol{\gamma}}^q)^{-1} 
    \Big ( \mathbf{D} (\boldsymbol{\gamma}_l^q) 
    \mathbb{E}_{p(\mathbf{x}, \mathbf{z})} \left [ \boldsymbol{\mu}_l^q \mathbf{x}^\mathrm{H} \right ] \nonumber \\
    & \qquad \qquad - \mathbf{D} (\tilde{\boldsymbol{\gamma}}_l^q) 
    \mathbb{E}_{p(\mathbf{x}, \mathbf{z})} \left [ \tilde{\boldsymbol{\mu}}_l^q \mathbf{x}^\mathrm{H} \right ]
    \Big ) \Bigg \} \Bigg ].
\end{align}

From lines 3-4 of Algorithm~\ref{alg:OvEP}, $\mathbb{E}_{p(\mathbf{x}, \mathbf{z})} \left [ \boldsymbol{\mu}_l^q \mathbf{x}^\mathrm{H} \right ]$ in \eqref{eq:align_x} can be calculated as
\begin{align}
    \label{eq:mu_x}
    \mathbb{E}_{p(\mathbf{x},\mathbf{z})} \left [ \boldsymbol{\mu}_l^q \mathbf{x}^\mathrm{H} \right ] 
    & = \sigma_z^{-1} \mathbf{\Sigma}_{l}^u \mathbf{H}_l^\mathrm{H} \mathbb{E}_{p(\mathbf{x},\mathbf{z})} \left [ \mathbf{y}_l \mathbf{x}^\mathrm{H} \right ] \nonumber \\
    & \overset{(a)}{=} \sigma_x \left ( \mathbf{H}_l^\mathrm{H} \mathbf{H}_l + \frac{\sigma_z}{\beta_x} \mathbf{I}_M \right )^{-1} \mathbf{H}_l^\mathrm{H} \mathbf{H}_l \nonumber \\
    & \overset{(b)}{=} \sigma_x \left \{ \mathbf{I}_M - 
    \left ( \frac{\beta_x}{\sigma_z }\mathbf{H}_l^\mathrm{H} \mathbf{H}_l + \mathbf{I}_M \right )^{-1} \right \} \nonumber \\
    & \overset{(c)}{=} \sigma_x \left ( \mathbf{I}_M - \beta_x^{-1} \mathbf{\Sigma}_l^u \right ),
\end{align}
where 
the equality $(a)$ follows from $\mathbb{E}_{p(\mathbf{x}, \mathbf{z})} \left [ \mathbf{y}_l \mathbf{x}^\mathrm{H} \right ] = \sigma_x \mathbf{H}_l$
and
$\bm{\gamma}_l^w = \tilde{\bm{\gamma}}_l^w = \beta_x^{-1} \mathbf{1}_M$ at the initial iteration $(t=1)$,
the equality $(b)$ follows from the identity $\left ( \mathbf{A}^\mathrm{H} \mathbf{A} + \alpha \mathbf{I}_M \right )^{-1} \mathbf{A}^\mathrm{H} \mathbf{A} = \mathbf{I}_M - \left ( \alpha^{-1} \mathbf{A}^\mathrm{H} \mathbf{A} + \mathbf{I}_M \right )^{-1} $, and
the equality $(c)$ follows from \eqref{eq:mu_sigma_u}.

Through the same derivation as \eqref{eq:mu_x}, $\mathbb{E}_{p(\mathbf{x}, \mathbf{z})} \left [ \tilde{\bm{\mu}}_l^q \mathbf{x}^\mathrm{H} \right ]$ in \eqref{eq:align_x} can be calculated as
\begin{align}
    \label{eq:mu_x_tilde}
    \mathbb{E}_{p(\mathbf{x},\mathbf{z})} \left [ \tilde{\boldsymbol{\mu}}_l^q \mathbf{x}^\mathrm{H} \right ] 
    = \sigma_x \left ( \mathbf{I}_M - \beta_x^{-1} \tilde{\mathbf{\Sigma}}_l^u \right ).
\end{align}
Substituting \eqref{eq:mu_x} and \eqref{eq:mu_x_tilde} into \eqref{eq:align_x} yields
\begin{align}
    \label{eq:align_x_2}
    &\sum_{l=1}^L \mathbb{E}_{p(\mathbf{x}, \mathbf{H}, \mathbf{z})} \left [ \mathbf{x}^\mathrm{H} \bar{\boldsymbol{\mu}}_l^q \right ] \nonumber \\
    &= \mathbb{E}_{p(\mathbf{H})} \Bigg [
    \sigma_x \mathrm{tr} \Bigg \{
        \underbrace{
    \mathbf{D}(\bar{\boldsymbol{\gamma}}^q)^{-1}  \sum_{l=1}^L \left (
    \mathbf{D}(\boldsymbol{\gamma}_l^q) - 
    \mathbf{D}(\tilde{\boldsymbol{\gamma}}_l^q)
    \right ) 
        }_{\overset{(a)}{=} \mathbf{I}_M}
    \Bigg \} \Bigg ] \nonumber \\
    &+\mathbb{E}_{p(\mathbf{H})} \Bigg [
    \sum_{l=1}^L \underbrace{ \mathrm{tr} 
    \Bigg \{ \mathbf{D}(\bar{\boldsymbol{\gamma}}^q)^{-1} \left (
    \mathbf{D}(\tilde{\boldsymbol{\gamma}}_l^q) \tilde{\mathbf{\Sigma}}_l^u - \mathbf{D}(\boldsymbol{\gamma}_l^q) \mathbf{\Sigma}_l^u \right ) \Bigg \} 
    }_{\overset{(b)}{\simeq} 0} \Bigg ] \nonumber \\
    & \simeq \sigma_x M,
\end{align}
where the equality $(a)$ follows from line 13 of Algorithm~\ref{alg:OvEP}.
The equality $(b)$ in \eqref{eq:align_x_2} follows from 
\begin{align}
    \mathrm{tr} & \left \{ \mathbf{D}(\bar{\boldsymbol{\gamma}}^q)^{-1} \left ( \mathbf{D}(\tilde{\boldsymbol{\gamma}}_l^q) \tilde{\mathbf{\Sigma}}_l^u - \mathbf{D}(\boldsymbol{\gamma}_l^q) \mathbf{\Sigma}_l^u \right ) \right \} \nonumber \\
    &= \mathbf{1}_M^\mathrm{T} \left \{ \left ( \tilde{\boldsymbol{\gamma}}_l^q \odot \mathbf{d}( \tilde{\mathbf{\Sigma}}_l^u) - {\boldsymbol{\gamma}}_l^q \odot \mathbf{d}({\mathbf{\Sigma}}_l^u) \right ) \oslash \bar{\boldsymbol{\gamma}}^q \right \} \nonumber \\
    & \overset{(c)}{\simeq} \mathbf{1}_M^\mathrm{T} \left \{ \left ( \tilde{\boldsymbol{\gamma}}_l^u \odot \mathbf{d}(\tilde{\mathbf{\Sigma}}_l^u) - {\boldsymbol{\gamma}}_l^u \odot \mathbf{d}({\mathbf{\Sigma}}_l^u) \right ) \oslash \bar{\boldsymbol{\gamma}}^q \right \} \nonumber \\
    &= \mathbf{1}_M^\mathrm{T} \left \{ \left ( \mathbf{0}_M \right ) \oslash \bar{\boldsymbol{\gamma}}^q \right \} = 0,
\end{align}
where the equality $(c)$ follows from 
$\bm{\gamma}_l^q = \bm{\gamma}_l^u - \bm{\gamma}_l^w \simeq \bm{\gamma}_l^u$ when $\bm{\gamma}_l^u \gg \bm{\gamma}_l^w$.
Finally, substituting \eqref{eq:align_x_2} into \eqref{eq:MSE_in}, $\mathrm{MSE} (\bar{\bm{\mu}}^q)$ can be expressed as \eqref{eq:MSE_in_apx}.

\section{Proof of Proposition~\ref{proposition:fixed}}
\label{apx:fixed}

Let $\{ \bm{\alpha}_l, \bm{\beta}_l \}_{l=0}^{L}$ and $\{ \tilde{\bm{\alpha}}_l, \tilde{\bm{\beta}}_l \}_{l=1}^{L-1}$ denote the Lagrange multipliers corresponding to the optimization problem with equality constraints defined in \eqref{eq:opt_moment}.
Then, the Lagrangian $\mathcal{L}$ can be expressed as 
\begin{align}
    \label{eq:Lagrangian}
    &\mathcal{L} ( \{b_l \}_{l=0}^L, \{\tilde{b}_l \}_{l=1}^{L-1}, q,
    \{ \bm{\alpha}_l, \bm{\beta}_l \}_{l=0}^{L}, 
    \{ \tilde{\bm{\alpha}}_l, \tilde{\bm{\beta}}_l \}_{l=1}^{L-1} ) \nonumber \\ 
    &= J \left ( \{b_l \}_{l=0}^L, \{\tilde{b}_l \}_{l=1}^{L-1}, q \right ) \nonumber \\
    &+ \sum_{l=0}^L \bm{\alpha}_l^\mathrm{T} \left \{  
    \mathbf{d} \left ( \mathbb{E}_{b_l (\mathbf{x}) }[ \mathbf{x} \mathbf{x}^\mathrm{H} ] \right ) -
    \mathbf{d} \left ( \mathbb{E}_{q (\mathbf{x})}[\mathbf{x} \mathbf{x}^\mathrm{H} ] \right ) \right \} \nonumber \\
    &+ \sum_{l=1}^{L-1} \tilde{\bm{\alpha}}_l^\mathrm{T} \left \{  
    \mathbf{d} \left ( \mathbb{E}_{\tilde{b}_l (\mathbf{x}) }[ \mathbf{x} \mathbf{x}^\mathrm{H} ] \right ) -
    \mathbf{d} \left ( \mathbb{E}_{q (\mathbf{x})}[\mathbf{x} \mathbf{x}^\mathrm{H} ] \right ) \right \} \nonumber \\
    &+ 2 \sum_{l=0}^{L} \mathfrak{R} \left [ \bm{\beta}_l^\mathrm{H} \left ( \mathbb{E}_{b_l (\mathbf{x}) }[\mathbf{x}] - \mathbb{E}_{q (\mathbf{x})}[\mathbf{x}] \right )\right ] \nonumber \\
    &+ 2 \sum_{l=1}^{L-1} \mathfrak{R} \left [ \tilde{\bm{\beta}}_l^\mathrm{H} \left ( \mathbb{E}_{\tilde{b}_l (\mathbf{x})}[\mathbf{x}] - \mathbb{E}_{q (\mathbf{x})}[\mathbf{x}] \right )\right ].
\end{align}
The Lagrange multipliers are set to 
\begin{subequations}
\label{eq:Lagrange_multipliers}
\begin{align}
    \bm{\alpha}_0 &= \bar{\bm{\gamma}}^q,\ 
    \bm{\beta}_0 = - \bar{\bm{\gamma}}^q \odot \bar{\bm{\mu}}^q, \\
    \bm{\alpha}_l &= \bm{{\gamma}}_l^w,\ 
    \bm{\beta}_l = - \bm{{\gamma}}_l^w \odot \bm{\mu}_l^w, \\
    \tilde{\bm{\alpha}}_l &= - \bm{\tilde{\gamma}}_l^w,\ 
    \tilde{\bm{\beta}}_l = \bm{\tilde{\gamma}}_l^w \odot \tilde{\bm{\mu}}_l^w.
\end{align}
\end{subequations}
The validity of this choice of Lagrange multipliers in \eqref{eq:Lagrange_multipliers} can be confirmed at the end of this proof.

In what follows, the stationary point of the Lagrangian in \eqref{eq:Lagrangian} is derived using the functional derivative with respect to the PDFs $b_0(\mathbf{x}), b_l(\mathbf{x}), \tilde{b}_l(\mathbf{x}), q(\mathbf{x})$.

Focusing on $q(\mathbf{x})$, the Lagrangian can be expressed as 
\begin{align*}
    \mathcal{L} &= \mathrm{H}(q) - 
    \underbrace{
    \left ( \bm{\alpha}_0 + \sum_{l=1}^L \bm{\alpha}_l + \sum_{l=1}^{L-1} \tilde{\bm{\alpha}}_l \right )^\mathrm{T} }_{ \overset{(a)}{=} \bm{\gamma}_\mathrm{f} }
    \mathbf{d} \left (\mathbb{E}_{q(\mathbf{x})} \left [ \mathbf{x} \mathbf{x}^\mathrm{H} \right ] \right ) \nonumber \\
    &\quad -2 \mathfrak{R} \Bigg [ 
    \underbrace{
    \left ( \bm{\beta}_0 + \sum_{l=1}^L \bm{\beta}_l + \sum_{l=1}^{L-1} \tilde{\bm{\beta}}_l\right )^\mathrm{H} }_{ \overset{(a)}{=} - \bm{\gamma}_\mathrm{f} \odot \mathbf{x}_\mathrm{f}}
    \mathbb{E}_{q(\mathbf{x})} \left [ \mathbf{x}\right ] \Bigg ]
    +\mathrm{const} \nonumber \\
    &= - \mathbb{E}_{q(\mathbf{x})} \left [\ln q(\mathbf{x}) \right ] + \mathbb{E}_{q(\mathbf{x})} \left [ \| \mathbf{x} -\mathbf{x}_\mathrm{f} \|_{\bm{\gamma}_\mathrm{f}}^2 \right ] + \mathrm{const},
\end{align*}
where the equality $(a)$ holds using the fixed point \eqref{eq:fixed_point_2} and the Lagrange multipliers \eqref{eq:Lagrange_multipliers}.
Then, taking the functional derivative with respect to $q(\mathbf{x})$ yields
\begin{align*}
    \frac{\delta L}{\delta q(\mathbf{x})}
    = -\ln q(\mathbf{x}) - 1 
    + \| \mathbf{x} - \mathbf{x}_\mathrm{f} \|^2_{ \bm{\gamma}_\mathrm{f} }.
\end{align*}
Hence, it can be seen that the PDF $q(\mathbf{x})$ at the stationary point, which satisfies $\frac{\delta \mathcal{L}}{\delta q(\mathbf{x})} = 0, \forall \mathbf{x}$, is equivalent to \eqref{eq:q}.

Similarly, the functional derivative with respect to $b_0 (\mathbf{x})$, $b_l(\mathbf{x}),\ l \in \{1,\ldots, L\}$ and $\tilde{b}_l(\mathbf{x}),\ l \in \{1,\ldots, L-1\}$ can be calculated as 
\begin{align*}
    \frac{\delta \mathcal{L}}{\delta b_0(\mathbf{x})}
    &=\ln b_0 (\mathbf{x}) + 1 - \ln f_0 (\mathbf{x})
    + \| \mathbf{x} - \bar{\bm{\mu}}^q \|^2_{ \bar{\bm{\gamma}}^q}, \\
    \frac{\delta L}{\delta b_l(\mathbf{x})}
    &=\ln b_l(\mathbf{x}) + 1 - \ln f_l(\mathbf{x}) + \| \mathbf{x} - \boldsymbol{\mu}_l^w \|^2_{\boldsymbol{\gamma}_l^w}, \\
    \frac{\delta L}{\delta \tilde{b}_l(\mathbf{x})}
    &= -\ln \tilde{b}_l(\mathbf{x}) - 1 + \ln \tilde{f}_l(\mathbf{x})
    - \| \mathbf{x} - \tilde{\boldsymbol{\mu}}_l^w \|^2_{\tilde{\boldsymbol{\gamma}}_l^w}.
\end{align*}
From these derivatives, we can confirm that the PDFs $b_0(\mathbf{x})$, $b_l(\mathbf{x})$, and $\tilde{b}_l(\mathbf{x})$ at the stationary point satisfying $\frac{\delta \mathcal{L}}{\delta b_0(\mathbf{x})} = 0$, $\frac{\delta \mathcal{L}}{\delta b_l(\mathbf{x})} = 0$ and $\frac{\delta \mathcal{L}}{\delta \tilde{b}_l(\mathbf{x})} = 0, \forall \mathbf{x}$ are equivalent to \eqref{eq:b_0}, \eqref{eq:b_l}, and \eqref{eq:b_l_tilde}, respectively.

Finally, the means and variances of the PDFs $b_0, b_l, \tilde{b}_l$ are calculated.
Following the same procedure as in the derivations of \eqref{eq:mu_sigma_denoiser}, \eqref{eq:mu_sigma_u}, and \eqref{eq:mu_sigma_u_tilde}, 
the means and variances can be calculated as 
\begin{align*}
    &\mathbb{E}_{b_0(\mathbf{x})} [\mathbf{x}] \!=\! \hat{\mathbf{x}}, 
    && \!\!\!\!\! \mathbf{d} \left ( \mathbb{E}_{b_0(\mathbf{x})} [ (\mathbf{x} - \hat{\mathbf{x}}) (\mathbf{x} - \hat{\mathbf{x}})^\mathrm{H} ] \right ) \!=\! \mathbf{1}_M \oslash \hat{\bm{\gamma}}, \\
    %
    &\mathbb{E}_{b_l(\mathbf{x})} [\mathbf{x}] \!=\! \bm{\mu}_l^u,
    && \!\!\!\!\! \mathbf{d} \left ( \mathbb{E}_{b_l(\mathbf{x})} [ (\mathbf{x} - \bm{\mu}_l^u) (\mathbf{x} - \bm{\mu}_l^u)^\mathrm{H} ] \right ) \!=\! \mathbf{1}_M \oslash \bm{\gamma}_l^u, \\
    %
    &\mathbb{E}_{\tilde{b}_l(\mathbf{x})} [\mathbf{x}] \!=\! \tilde{\bm{\mu}}_l^u,
    && \!\!\!\!\! \mathbf{d} \left ( \mathbb{E}_{\tilde{b}_l(\mathbf{x})} [ (\mathbf{x} - \tilde{\bm{\mu}}_l^u) (\mathbf{x} - \tilde{\bm{\mu}}_l^u)^\mathrm{H} ] \right ) \!=\! \mathbf{1}_M \oslash \tilde{\bm{\gamma}}_l^u.
\end{align*}
From these results and the equality \eqref{eq:fixed_point} at the fixed point, 
the PDFs $b_0, b_l, \tilde{b}_l, q$ satisfy the moment-matching constraints in \eqref{eq:opt_moment}.
Consequently, the selection of the Lagrange multipliers in \eqref{eq:Lagrange_multipliers} is valid, since it satisfies the equality constraints.

\bibliographystyle{IEEEtran}
\bibliography{reference.bib}

\end{document}